\documentclass{article}

\usepackage{cite}
\usepackage{amsmath,amssymb,amsfonts}
\usepackage{algorithm}
\usepackage{algpseudocode}  % cleaner modern replacement for 'algorithmic'
\usepackage{caption}
\usepackage{graphicx}
\usepackage{textcomp}
\usepackage{xcolor}
\usepackage{amsthm}
\usepackage{pgfplots}
\pgfplotsset{compat=1.18}
\usepackage{tikz}
\usetikzlibrary{arrows.meta,positioning}
\usetikzlibrary{calc,through,intersections}
\usepackage{booktabs}
\usepackage{enumerate}
\usepackage{multirow}
\usepackage{subcaption}
\usepackage{hyperref}

\newcommand{\mcL}{\mathcal{L}}
\newcommand{\mcB}{\mathcal{B}}
\newcommand{\mcI}{\mathcal{I}}

\newcommand{\mbfA}{\mathbf{A}}

\newtheorem{theorem}{theorem}[section]

\newtheorem{lemma}[theorem]{Lemma}
\newtheorem{proposition}[theorem]{Proposition}
\newtheorem{corollary}[theorem]{Corollary}
\newtheorem{definition}[theorem]{Definition}
\newtheorem{remark}[theorem]{Remark}
\newtheorem{example}[theorem]{Example}

\title{Bipartiteness in Progressive Second-Price Multi-Auction Networks with Perfect Substitute}

\begin{document}

\maketitle
%\begin{center}
%\begin{minipage}{5cm}
\setlength{\parindent}{-1.3cm}
Jordana~Blazek~\href{https://orcid.org/0000-0003-0840-693X}{$^{1[0000-0003-0840-693X]}$}, 
and  Frederick~C~Harris,~Jr.~\href{https://orcid.org/0000-0002-0857-6931}{$^{2[0000-0002-0857-6931]}$}\\ \\
\small
\centering
University of Nevada, Reno, Nevada, USA 
%\end{minipage}
%\end{center}

\begin{abstract}
We consider a bipartite network of buyers and sellers, where the sellers run locally independent 
Progressive Second-Price (PSP) auctions, and buyers may participate in multiple auctions, 
forming a multi-auction market with perfect substitute. The paper develops a projection-based
influence framework for decentralized PSP auctions.
We formalize primary and expanded influence sets using projections on the active
bid index set and show how partial orders on bid prices govern allocation, market shifts, 
and the emergence of saturated one-hop shells. Our results highlight the robustness 
of PSP auctions in decentralized environments by introducing saturated components 
and a structured framework for phase transitions in multi-auction dynamics. 
This structure ensures deterministic coverage of the strategy space, enabling stable
and truthful embedding in the larger game. We further model intra-round dynamics
using an index $\tau_k$ to capture coordinated asynchronous seller updates coupled through
buyers' joint constraints. Together, these constructions explain how local 
interactions propagate across auctions and gives premise for coherent 
equilibria--without requiring global information or centralized control.
%\keywords{Progressive Second Price Auction (PSP), Influence Sets, 
%Monotone Convergence, Decentralized Markets}
\end{abstract}

%%%%%%%%%%%%%%%%%%%%%%%%%%%%%%%%%%%%%%%%%%%%%%%%%%%%%%%%%%%%%%%%%%%%%%%%

\section{Introduction}

%%%%%%%%%%%%%%%%%%%%%%%%%%%%%%%%%%%%%%%%%%%%%%%%%%%%%%%%%%%%%%%%%%%%%%%%

The Progressive Second Price (PSP) auction, introduced by Lazar and Semret~\cite{Lazarsemret1998}, 
and later expanded upon in Semret’s dissertation~\cite{Semret2000}, presented a full 
theoretical framework for distributed resource pricing and demonstrated the linkage between 
PSP and VCG‑type efficiency results. The PSP auction is a decentralized mechanism
characterized by truthfulness, individual rationality, and social welfare maximization. 
Unlike traditional centralized auctions, PSP allows buyers and sellers to iteratively 
interact through local bidding rounds, dynamically allocating consumable resources such as 
network bandwidth and other communicative and computational resources. 
In PSP auctions, winners pay a cost determined by the externality that is imposed on others,
calculated from the distribution of allocations and the bid, This ensures truthful 
reporting of valuation 
through incentive compatibility as was shown in the foundational work of Vickrey,
Clarke, and Groves~\cite{Vickrey1961, Clarke1971, Groves1973}.
The resulting equilibria adhere to the exclusion-compensation principle, 
preventing unilateral improvement without harming another participant.

Our focus is on developing adaptive auction mechanisms, like the Progressive 
Second Price (PSP) auction, that respond to market dynamics by allowing agents to 
adjust their bids based on local information gathered from their network neighbors.
This motivates the study of influence sets, dynamic participation, and the role 
of network effects in shaping bidding behavior. In these settings, agents 
lack full market information and are affected by network dependencies.

Maillé et al.~\cite{Maille2004} build directly on Lazar and Semret’s 1999 PSP framework by 
addressing the one remaining free parameter in the model — the reserve price.
They demonstrate that while PSP guarantees convergence, efficiency, and incentive compatibility,
the seller’s reserve price can be optimized by simple 
numerical methods, allowing PSP markets to balance efficiency with revenue maximization.

These iterative updates operate as strategic interactions in a decentralized framework, 
where the PSP auction converges, perhaps astonishingly, deterministically to an 
$\varepsilon$-Nash equilibrium. This has been shown to be true on the networks of
20 years ago, when bandwidth and bandwidth allocation was perhaps a different game.
A real-world, modern network faces significant obstacles; it is a game of partial
information played in a web of interconnected decisions, dynamic participation, 
and evolving market constraints. This motivates a graph-theoretic treatment of 
information flow and motivates the introduction of the concept of market saturation.

The structure of this paper is as follows. Section~\ref{sec:background} introduces
the foundations necessary to model influence propagation in decentralized auctions. 
In Section~\ref{sec:mechanism}, we present the Progressive Second Price (PSP)
auction mechanism, outlining its bidding rules, participation logic, 
and price allocation behavior. Section~\ref{sec:inf-sets} defines and explores 
the dynamics of influence sets, 
establishing a framework for analyzing how strategy updates propagate acrossthe
market. Our approach adopts and extends these concepts through graph-based methods, 
specifically leveraging the bipartite graph to systematically represent
buyer--seller interactions. Section~\ref{sec:saturation} introduces the concept 
of saturation as the limit of influence propagation, characterizing a locally 
evolving equilibrium structure. The simulation framework and implementation
are discussed in Section~\ref{sec:simulation-framework}, and this paper's 
conclusion and future work are presented in Section~\ref{sec:conclusion}.

%%%%%%%%%%%%%%%%%%%%%%%%%%%%%%%%%%%%%%%%%%%%%%%%%%%%%%%%%%%%%%%%%%%%%%%%

\section{Background and Related Work}\label{sec:background}

%%%%%%%%%%%%%%%%%%%%%%%%%%%%%%%%%%%%%%%%%%%%%%%%%%%%%%%%%%%%%%%%%%%%%%%%

This paper introduces a graph-based analytical framework to examine the dynamics of 
Progressive Second Price (PSP) auctions within decentralized market structures. 
Our approach builds on foundational concepts in auction theory, network influence
propagation, and graph analysis, while situating the PSP model among several related domains.

The Progressive Second Price (PSP) auction, initially proposed by Lazar 
and Semret~\cite{Lazarsemret1998}, extends classical second-price 
mechanisms~\cite{Vickrey1961, Clarke1971, Groves1973} into decentralized contexts. 
Earlier studies such as Maille and Tuffin~\cite{Maille2007} and Semret’s 
dissertation~\cite{Semret2000} provided a full system-level model of distributed 
market control and the theoretical grounding for the PSP auction mechanism, 
analyzing network-based PSP equilibria and pricing strategies. Subsequent 
formal analyses such as Qu, Jia, and Caines~\cite{Qujia2007} presented key 
results on the Uniformly Quantized PSP (UQ‑PSP) mechanism, showing that it 
guarantees convergence to a unique limit price independent of initial conditions, 
achieves $\gamma$‑incentive compatibility, and extends naturally to network 
topologies where equilibria depend on local information exchange. Their framework 
provided the first rigorous quantized extension of the PSP model, establishing
discrete convergence proofs that later generalizations such as those 
developed in this work,
Qu, Jia, and Caines~\cite{Qujia2009} further extended these results to networked PSP 
convergence, introducing asynchronous coordination and bounded-delay convergence.
Subsequent work has investigated distributed or multi-resource variants, 
including privacy-preserving and differential frameworks in data and spectrum
markets~\cite{Brandt2008, Wang2021}, expanding PSP-like mechanisms to new computational settings.

Local coordination rules, when combined with bounded delays and limited information 
exchange, can achieve global properties similar to those in consensus
and averaging protocols. Beyond traditional equilibrium analysis, distributed 
consensus and coordination models offer insight into asynchronous bidding and update rules.
Aguilera and Toueg~\cite{Aguilera2012} and Lynch~\cite{Lynch1996} describe protocols
ensuring eventual consistency under partial information, concepts that are applicable to
asynchronous PSP updates. These works demonstrate that convergent systems operating under
bounded delay result in deterministic convergence guarantees in decentralized markets.

In decentralized markets, agents’ strategies depend on local interactions but propagate
indirectly through shared participation and local coordination. This connects PSP analysis to the 
broader literature on influence diffusion and cascading behavior, as in 
Kleinberg~\cite{kleinberg2007}, Oki et al.~\cite{Oki2018}, and 
Osvaldo and Queen~\cite{Osvaldo2017}, which examine network-driven contagion 
and adaptive decision processes. The theoretical foundation of influence sets 
also aligns with the study of sphere-of-influence graphs~\cite{Quint1994, Toussaint2014TheSO} 
and dynamic graph structures that represent iterative strategic dependencies.

Graph-based approaches are central to understanding multi-agent optimization. 
Baur~\cite{Baur2012}, Barrett~\cite{Barrett2017}, and related work on planar 
and dynamic graphs illustrate how reachability, closure, and resistance distance can
capture evolving connectivity. In the PSP context, our use of projection operators 
extends these methods by linking graph reachability to economic stability, 
enabling a deterministic interpretation of market influence propagation.

This work examines decentralized auction theory, distributed coordination,
influence propagation, and graph-theoretic modeling to provide a coherent 
analytical framework for PSP auctions. This expanded foundation motivates the 
later sections on local saturation and asynchronous market dynamics.

%%%%%%%%%%%%%%%%%%%%%%%%%%%%%%%%%%%%%%%%%%%%%%%%%%%%%%%%%%%%%%%%%%%%%%%%
\section{The PSP Auction Mechanism}\label{sec:mechanism}
%%%%%%%%%%%%%%%%%%%%%%%%%%%%%%%%%%%%%%%%%%%%%%%%%%%%%%%%%%%%%%%%%%%%%%%%

The Progressive Second Price (PSP) auction is a decentralized mechanism in which
buyers iteratively submit bids to sellers, and sellers update reserve prices 
based on received bids. Each auction operates locally, and coordination emerges 
through repeated interactions across the market graph. 
The mechanism rules first appears in~\cite{Lazarsemret1998}, defining the bid 
structure, auction dynamics, pricing rules, allocation strategies, 
and participation behavior. In what follows, we define the bid structure, 
auction dynamics, pricing rules, allocation strategies, and participation 
behavior that govern the PSP mechanism.
%%%%%%%%%%%%%%%%%%%%%%%%%%%%%%%%%%%%%%%%%%%%%%%%%%%%%%%%%%%%%%%%%%%%%%%%
Let $\mcI = \mcB \cup \mcL$ denote the set of all agents, partitioned into 
buyers and sellers. Each seller $j \in \mcL$ manages a local auction for
a divisible resource, and each buyer $i \in \mcB$ may submit bids to a
subset of sellers. The bid profile of auction $j$ is given by the column 
vector $s^j$ with entries $s_i^j$, where $(i,j) \in \mcB \times \mcL$. A bid
\[
s_i^j=(q_i^j, p_i^j)\in S_i^j = [0,Q^j]\times[0,\infty)
\]
represents a single interaction between buyer $i$ and seller $j$, 
where $q_i^j$ is the quantity requested by the buyer and $p_i^j$ is the unit price offered.

In decentralized markets governed by distributed Progressive Second Price (PSP) auctions,
agents submit bids in the form of price-quantity pairs at discrete time steps. 
These bids are locally observable: buyers receive feedback from auctions in which 
they participate, and sellers observe aggregate demand over time. However, 
the global structure of the market--including overlapping buyer influence, 
competition externalities, and inferred network effects--must be reconstructed 
from these partial, temporally indexed signals. 

\begin{table}[h]
\centering
\renewcommand{\arraystretch}{1.2}
\begin{tabular}{@{} l | c | c @{}}
\textbf{Object} & \textbf{Single auction $j$} & \textbf{Across all auctions} \\ \hline
quantity & $Q^{j}$ & $\bigl(Q^{1},\dots,Q^{J}\bigr)$ \\[2pt]
Player $i$’s bid pair & $s_{i}^{j}=(q_{i}^{j},\ p_{i}^{j})$ & $s_{i}=(s_{i}^{1},\dots,s_{i}^{J})$ \\[2pt]
Strategy space of player $i$ & $S_{i}^{j}=[0,Q^{j}]\times[0,\infty)$ & $S_{i}=\displaystyle\prod_{j=1}^{\mcL} S_{i}^{j}$ \\[4pt]
\emph{Opposing bids w.r.t.\ player $i$} & $s_{-i}^{j}=(s_{1}^{j},\dots,s_{i-1}^{j},s_{i+1}^{j},\dots,s_{n}^{j})$ & $s_{-i}=(s_{-i}^{1},\dots,s_{-i}^{J})$ \\[4pt]
Profile in auction $j$ & $s^{j}=(s_{1}^{j},\dots,s_{n}^{j})$ & \multirow{2}{*}{$s=(s^{1},\dots,s^{J})$} \\[2pt]
Grand strategy space & $S^{j}=\displaystyle\prod_{i=1}^{n} S_{i}^{j}$ & $S=\displaystyle\prod_{j=1}^{\mcL} S^{j}$ \\
\end{tabular}
\caption{Basic sets and notation for a bundle of $J$ independent PSP auctions}
\end{table}

%––––––––––––––––––––––––––––––––––––––––––––––––––––––––
\subsection{Bounded Participation}

Each buyer will know the available quantity for each market in which they bid. 
Buyers act strategically by selecting sellers, adjusting bid quantities,
and choosing whether to participate based on their expected ability to 
satisfy demand. In the PSP framework buyers cannot reveal their entire valuation functions in a
single step; instead they must request allocations iteratively. To regulate this
behavior we introduce a bounded participation rule, which endogenously limits
the set of sellers a buyer engages with, and can be seen as an analogue of the opt-out behavior
given in~\cite{Blocher2021}.

Fix buyer $i$ at time $t$ and let $p^*$ denote the common marginal price
identified from opponents’ bids. For each seller $j$
let $c_j=\mathrm{cap}_j(p^*)$ be the residual quantity available to $i$ at price
$p^*$. Define the desired total quantity
\begin{equation}\label{eq:total_z}
  z_i^* = \min\bigg\{\bar q_i(t), \ \sum_{j} c_j\bigg\}.
\end{equation}

\begin{definition}[Bounded participation rule]\label{def:bounded_participation}
Buyer $i$ selects a minimal--cost subset of sellers
$\mathcal{L}_i(t)\subseteq\mathcal{L}$, ordered by nondecreasing price $p_{(n)}^j(t)$,
such that
\begin{equation}\label{eq:min-cost-sellers}
  \sum_{j\in\mathcal{L}_i(t)} c_j(t) \ge z_i^*.
\end{equation}
The buyer allocates requests sequentially to the least expensive sellers until
the desired total quantity $z_i^*$ is reached, subject to residual capacities
$c_j(t)$. For $j\notin\mathcal{L}_i(t)$, set $q_i^j=0$.
\end{definition}

This rule formalizes bounded participation at fixed $t$: each buyer interacts 
only with the fewest necessary sellers to realize $z^*$, in an attempt
to minimize the cost of participation. The resulting allocation 
targets allocations at a common marginal price $p^*(t)$ under residual 
quantity constraints.

\subsection{Residual Quantity and Allocation}

As a market with perfect but incomplete information, sellers can only gain
information about demand by observing buyer behavior, 
determined by the connectivity of the auction graph.
In each iteration, every seller 
completes one update of its local auction.

For each seller $j$, the reserve price $p_*^j(t)$ is the
price at which seller~$j$ is indifferent between selling her final
unit of resource and keeping it. Equivalently, the seller may be
viewed as submitting an internal bid $(Q^j, p_*^j(t))$ on her own
auction.
At the end of each round $t$, the reserve price is updated with
information from the set of active bids,
where $\mcB^j(t)$ is the set of buyers who win strictly positive
allocations at seller~$j$ in round~$t$, and $\epsilon>0$.

We define the clearing price at seller~$j$ to be the smallest price at which
aggregate awarded quantity meets available quantity:
\begin{equation}\label{eq:clearing_price}
  \chi^j(t) = \min\bigg\{y :\sum_{k:\,p_k^j(t)>y} a_k^j(t) \ge Q^j(t)\bigg\}.
\end{equation}
Any residual supply must therefore be allocated among bids that tie
at prices just above $\chi^j(t)$, after higher–priced bids are filled.
Let
\begin{equation}\label{eq:margins}
  \underline p^j(t) := \min\{\,p_i^j(t) : i \in \mcB^j(t)\,\},
  \qquad
  \overline p^j(t) := \max\{\,p_i^j(t) : i \not\in \mcB^j(t)\,\},
\end{equation}
be the lowest winning and highest losing bid prices at seller~$j$, and
where buyers \emph{not} in $\mcB^j(t)$ receive zero allocation at seller~$j$.
The clearing price satisfies
\[
  \overline p^j(t) < \overline p^j(t) + \epsilon \le
  \chi^j(t) \le \underline p^j(t) - \epsilon < \underline p^j(t)
\]
whenever there is at least one winning and one losing bidder at
seller~$j$. In particular, $\chi^j(t)$ lies in the open interval between
the highest losing and lowest winning bid.
At equilibrium, the reserve price $p_*^j(t)$ coincides with the clearing
price at seller~$j$, i.e., the clearing price implied by the PSP
allocation rule.

Buyers at higher prices are therefore always served in full,
whereas buyers at the threshold price may be rationed.  
At each price level $y$, the residual quantity is given by
\begin{equation}\label{eq:residual_R}
    R^j(y,t) = \Big[ Q^j(t) - \sum_{k:p_k^j(t)>y} a_k^j(t) \Big]^+.
\end{equation}

When multiple buyers tie at $p_i^j(t)=y$, the awarded allocation respects both the 
buyer’s request and the residual supply. We refer to the tie–splitting rule originated 
in the analysis of quantized PSP auctions by Qu, Jia, and Caines~\cite{Qujia2007},
\begin{equation}\label{eq:allocation-prop-matrix}
    a_i^j(s(t)) = \min\bigg\{ q_i^j(t), \frac{q_i^j(t)}{\sum_{\ell:p_\ell^j(t)=y} q_\ell^j(t)} R^j(y,t) \bigg\}.
\end{equation}

The bid quantity $q_i^j(t)$ and the allocation $a_i^j(t)$
are complementary. In fact, the buyer strategy is the first term in the
minimum, the second term being owned by the seller. 
For each buyer--seller pair $(i,j)$ at time $t$, $a_i^j(t)$ is the
\emph{awarded} amount that seller $j$ allocates to buyer $i$ once the allocation 
rule has been applied. By construction,
\begin{equation}\label{eq:awarded_matrix}
    a_i^j(t) \le q_i^j(t),
\end{equation}
with equality holding when residual supply at the
buyer’s price suffices to satisfy 
all tied requests. The mechanism therefore never awards 
more than requested and may 
award less when quantity is limited.

We remark that the reserve price $p_*^j(t)$ that lies in the margin interval determined
by the bids
\begin{equation}\label{eq:reserve_interval}
  \overline p^j(t) < p_*^j(t) < \underline p^j(t),
\end{equation}
whenever both $\overline p^j(t)$ and $\underline p^j(t)$ are defined,
and we deliberately leave the precise rule for
selecting $p_*^j(t)$ within the interval~\eqref{eq:reserve_interval}
unspecified. In particular, admissible choices include
\[
  p_*^j(t) = \chi^j(t),\qquad
  p_*^j(t) = \overline p^j(t) + \epsilon,\qquad
  p_*^j(t) = \underline p^j(t) - \epsilon,
\]
provided that reserve price updates lie within $\epsilon$ and the resulting
sequence $\{p_*^j(t)\}_t$ is nondecreasing.

\subsection{Exclusion--Compensation}

Each buyer’s payment follows a second–price externality principle, 
this is the ``social opportunity cost'' of the PSP pricing rule.
The exclusion–compensation payment to buyer~$i$ equals the 
loss imposed on other buyers at that seller when $i$ participates.
For a fixed auction $j$ we use the opposing buyers' 
piecewise--constant marginal price function
$P^j(\cdot, s_{-i}^j)$ built from $s_{-i}^j$,
\begin{equation}\label{eq:const_cost}
    c_i^j(s) = \int_0^{a_i^j(s)} P^j\big(z, s_{-i}^j\big) dz,
\end{equation}
which holds true locally at each auction, where the opposing bids are
calculated against the allocated resource to buyer $i$.
The amount of resource available at price $p^j_{(n)}$ is
$\xi^j_{n-1} - \xi^j_n \ge 0$. The local inverse price function is then
\begin{equation}\label{eq:local-inv-price}
  P^j(z,s_{-i}^j) = p^j_{(n)}\quad \text{for }z\in(\xi^j_n,\xi^j_{n-1}].
\end{equation}
For each ordered price $y$, we have that $P_i(z, s_{-i})$ is defined
for the range of $z$ corresponding to the total resource
available from all sellers at that price, i.e.,
\begin{equation}\label{eq:zrange}
  z \in \bigg(\sum_{p_{(n)}^j>y}
                (\xi^j_{n-1} - \xi^j_n),\
              \sum_{p_{(n)}^j\ge y}
                (\xi^j_{n-1} - \xi^j_n)\bigg].
\end{equation}
Define the aggregate residual quantity
\begin{equation}\label{eq:agg-residual}
    Q_i(y, s_{-i}) \ =\ \sum_{j=1}^\mcL Q_i^j(y, s_{-i}^j),
  \qquad
  P_i(z, s_{-i}) \ =\ \inf\{\ y\ge 0:\ Q_i(y, s_{-i}) \ge z\ \},
\end{equation}
where because $Q_i(y, s_{-i})$ is a right--continuous, nondecreasing step function
with finitely many jumps at $\{p_{(m)}^j\}$, the infimum is attained.

\subsection{Valuation and Utility}

Each buyer~$i$ has an elastic valuation function $\theta_i:[0,Q_i]\to[0,\infty)$ 
with strictly decreasing derivative $\theta_i'$. The valuation depends on the total 
awarded quantity across all sellers:
\begin{equation}\label{eq:buyervaluation}
    V_i(a) = \theta_i \bigg(\sum_{j=1}^J a_i^j(t)\bigg)
    = \int_0^{\sum_{j=1}^J a_i^j(t)} \theta'_i(z) dz.
\end{equation}
Given a strategy profile $s$, the utility of buyer $i$ for potential 
allocation $a$ is dependent on the cost, $c_i(s)$, where
the cost to buyer $i$ as a function of the entire strategy profile $s$.
%In a multi--auction
%setting we consider the realization of an allocation at time $t$
%assuming a zero--delay environment, whereas Lazar and Semret originally analyzed
%asynchronous updates where bids may arrive at different times. 

In the dynamic setting this profile evolves with iteration $t$, 
where $c_i(s)$ may represent total participation costs, including membership fees, 
per–round overhead, and per–auction message costs.
Utility is given by
\begin{equation}\label{eq:utility}
    u_i(s) = V_i(a) - c_i(s),
\end{equation}
where $c_i(s)$ is a dynamic cost function that evolves over time with bid updates.

The buyers’ utility functions
implicitly define a potential over the allocation space, as buyers seek to
maximize their utility through strategic allocation requests.
We note that a uniform (coordinated) bid price from buyers across active sellers upholds
strategic simplicity and second-price incentives, which are rational under 
quasi-linear utilities, as shown in the original PSP framework~\cite{Lazarsemret1998}.

Following Lazar and Semret~\cite{Lazarsemret1998}, updates occur only
when the buyer’s utility improvement exceeds a small positive
threshold, ensuring asynchronous convergence under bounded delay.
In a single–auction market, buyer $i$ accepts a new bid $s_i'$ only if
$ u_i(s_i'; \ s_{-i}) - u_i(s_i;\ s_{-i}) \ >\ \varepsilon$.
In the multi–auction setting, buyer $i$ posts a vector of bids 
that share a common marginal price $p_i^*$ across all connected sellers.
The utility comparison therefore becomes an aggregate test, where,
in terms of the opposing bid vector $s_{-i}$,
any gain in utility at time $t$ depends on the current state of play.  
Information propagation across the market affects how the vector of opposing 
bids $s_{-i}$ is formed, and thus how externalities are computed. 
The realized utility 
improvement $\Delta u_i(t)$ is evaluated relative to the previous round to 
determine if a new bid exceeds the cost of participation.

The discussion of externality under multiple auctions running asynchronously and 
a formal convergence analysis of the PSP mechanism to a single, unique, global $\varepsilon$-Nash 
network equilibrium, as was given in~\cite{Semret2000}, is outside the scope of this paper.
We instead focus on the iterative application of a uniform marginal price and the 
localized pricing structure resulting from progressive bid updates on connected 
network components consisting of multiple sellers sharing multiple buyers under
an assumed bipartite structure. A formal analysis of the effects of latency on a 
PSP auction is given in~\cite{Blazek2025}.

%%%%%%%%%%%%%%%%%%%%%%%%%%%%%%%%%%%%%%%%%%%%%%%%%%%%%%%%%%%%%%%%%%%%%%%%
\section{Influence Sets}\label{sec:inf-sets}                                 
%%%%%%%%%%%%%%%%%%%%%%%%%%%%%%%%%%%%%%%%%%%%%%%%%%%%%%%%%%%%%%%%%%%%%%%%

We model the behavior of vertices (buyers and sellers) in this bipartite structure using 
influence sets. Each vertex’s strategy space is influenced by neighboring vertices and 
evolves over time. Influence sets restrict the strategy space of buyers and sellers 
within bounded regions, stabilizing auction dynamics and supporting predictable market 
equilibria. As rational agents, buyers and sellers do not optimize perfectly but instead
operate within acceptable thresholds of cost and utility. Influence propagation determines
the flow of information, where bidding saturation occurs once influence sets stabilize.
At saturation, there is no vertex that, upon calculating his measure of utility, suffers
a changing set of opponent bids, and all subsequent bid updates are calculated on locally
stable subgraphs of the market.

Thus, influence sets are subsets of the auction graph that represent the scope of influence a 
particular vertex (buyer or seller) has on others over a finite number of auction iterations.
These sets structure interactions into subsets of the auction graph where local equilibria
form dense regions where bid updates are stabilized.

%–––––––––––––––––––––––––––––––––––––––––––––––––––––––––
\subsection{Primary (Direct) Influence Sets}

Following the original definition from~\cite{Blocher2021}, the \emph{primary influence set}, 
denoted $\Lambda$, for a given seller~$j$ at time~$t$, is defined set-theoretically,
\begin{equation}\label{eq:seller-direct} 
\Lambda_{\mcL}(j,t) = \bigcup_{i\in\mcB^j(t)} \mcL_i(t), 
\end{equation} 
where $\mcB^j(t)$ is the set of buyers bidding on seller~$j$, and $\mcL_i(t)$ is the 
set of sellers that buyer~$i$ bids on. Thus, $\Lambda^{(1)}_{\mcL}(j,t)$ represents all
sellers directly connected to auction~$j$ via shared buyers at time~$t$. This definition 
captures the notion that influence propagates across the auction graph through 
buyer--seller connections. The superscript $(1)$ denotes the first layer of influence 
anchored at seller~$j$ expanding through buyer-mediated connections.

To illustrate, for a buyer~$i$, the relevant bids $s_i^j$ flow \emph{from} the buyer 
to sellers. For a seller~$j$, we reverse this; bids flow \emph{into} the seller from 
buyers. The base case captures this directionality, which we get from market theory: 
buyers have positive demand, and sellers have negative demand (otherwise known as surplus). 
The direct influence set for buyer~$i$, denoted $\Lambda_{\mcB}^{(1)}(i,t)$, includes
buyers directly connected to~$i$ through shared sellers,
\begin{equation}\label{eq:buyer-direct}
    \Lambda_{\mcB}(i,t) = \bigcup_{j \in \mcL_i(t)} \mcB^j(t).
\end{equation}

\noindent
We now have the first layer of \emph{buyer-to-buyer} influence induced by common seller
participation. It serves as a foundation for constructing buyer--buyer influence 
graphs and identifying bid coordination structures within the network. This expression
gathers the buyers indirectly connected to buyer~$i$ through shared sellers, filtered 
by active bids at the given iteration. It provides a way to trace buyer--buyer 
influence mediated through seller auctions.

We extend the definition from~\cite{Blocher2021} in a theoretical and practical sense, 
defining the base case explicitly as the vertex itself,
\begin{equation}\label{eq:Lambda-base}
\Lambda^{(0)}(x,t) = \{ x\},
\end{equation}
\noindent
emphasizing that at the zeroth level, the influence set represents only the vertex 
itself. This represents a measure of ``self--influence'', such as reserve prices
(for sellers) or initial valuations (for buyers), economically aligning with the
idea that a seller starts from a reserve price reflecting their own valuation, 
while a buyer’s self-valuation corresponds to their initial maximum willingness-to-pay.

%–––––––––––––––––––––––––––––––––––––––––––––––––––––––––
\subsection{Expanded (Indirect) Influence Sets}
For any vertex $x$ in the auction graph (buyer or seller), and for $n\ge 1$,
the primary influence set is expanded from the $(n-1)$--step influence set by 
aggregating direct neighbors at the next layer:
\begin{equation}\label{eq:Lambda-recursive}
\Lambda^{(n)}(x,t) = \bigcup_{y \in \Lambda^{(n-1)}(x,t)} \Lambda^{(1)}(y,t). 
\end{equation}

\noindent
Where we define a two--hop projection operator
\[
\Lambda^{(1)}(y,t):=
\begin{cases}
\displaystyle\bigcup_{i\in\mcB^{y}(t)}\mcL_i(t), & \quad \text{if} \quad y\ \in\ \mcL,\\
\displaystyle\bigcup_{j\in\mcL_{y}(t)}\mcB^{j}(t), & \quad \text{if} \quad y\ \in\ \mcB,
\end{cases}
\]
which always returns vertices of the \emph{same type} as $y$ after one buyer--seller alternation.
For sellers, the $n$--step influence set may be computed recursively,
\[
\Lambda^{(n)}_{\mcL}(j,t) = \bigcup_{i \in \Lambda^{(n-1)}_{\mcB}(j,t)} \mcL_i(t),
\]
where $\Lambda^{(n-1)}_{\mcB}(j,t)$ is the set of buyers reachable from seller~$j$ 
in $n-1$ steps. This returns the set of sellers that receive nonzero bids from buyers 
who are indirectly connected to auction~$j$ via shared bidding activity across 
$n$ rounds of the PSP auction. It describes how seller~$j$’s influence propagates 
through buyer behavior across seller neighborhoods.
For buyers,
\[
\Lambda^{(n)}_{\mcB}(i,t) = \bigcup_{j \in \Lambda^{(n-1)}_{\mcL}(i,t)} \mcB^j(t),
\]
where $\Lambda^{(n-1)}_{\mcL}(i,t)$ collects sellers indirectly connected to buyer~$i$.

Each new layer $\Lambda^{(n)}(x,t)$ therefore adds the direct neighbors of all vertices
in the previous layer, producing a breadth–first expansion in the auction graph.
This recursive expansion therefore builds a ``growing influence ball” centered at $x$,
where the secondary set acts as a generalized neighborhood closure or hull
around the initial primary set. At each step $n$, the influence set $\Lambda^{(n)}(x,t)$ 
forms an outer boundary surrounding the influence set $\Lambda^{(1)}(x,t)$, 
recursively aggregating direct neighborhoods around previously identified influence vertices.

%–––––––––––––––––––––––––––––––––––––––––––––––––––––––––
\paragraph*{Pathwise Characterization.}
In graph theory, this structure parallels the $n$-hop neighborhood closure or a 
breadth-first expansion of distance-$n$ shells. We characterize $\Lambda^{(n)}(x,t)$ 
pathwise as the set of all vertices reachable from $x$ by paths alternating between 
buyers and sellers, of length up to $2n$. Formally, let $\mathcal{G} = (\mcI, E)$ 
denote the bipartite auction graph, where $\mcI$ is the set of agents 
(buyers and sellers), and an edge $(i, j) \in E$ exists if buyer~$i$ bids on seller~$j$. 
The graph alternates between buyers and sellers by construction: no two buyers or two 
sellers are directly connected. Because of the bipartiteness we have the parity rule, and consequently
\[
\Lambda^{(n)}(x,t)=
\bigl\{y\in\mcI \ \bigl|\  \operatorname{dist}_{\mathcal G}(x,y)=2n\bigr\}.
\]

From a strategic perspective, the expanded influence set $\Lambda^{(n)}(x,t)$ 
describes the scope of anticipated externalities: the agents whose actions may 
not affect $x$ directly, but may impact $x$'s incentives via shared neighbors.
These influence chains emerge in environments with incomplete information 
and approximate the region of the market that affects the \emph{expected utility gradient}
of vertex $x$. In equilibrium analysis, these indirect sets are crucial for 
understanding stability, coordination potential, and susceptibility to shock 
propagation (e.g., strategic manipulation or correlated noise).
As noted in~\cite{Blocher2021} and echoed in broader decentralized market theory 
(e.g.,~\cite{shah2011message}), indirect influence plays a key role in 
shaping convergence. While $\Lambda(x,t)$ governs observed interaction, 
$\Lambda^{(n)}(x,t)$ governs inferred or mediated interdependence--and together, 
they define the full strategic visibility of a vertex.

%%%%%%%%%%%%%%%%%%%%%%%%%%%%%%%%%%%%%%%%%%%%%%%%%%%%%%%%%%%%%%%%%%%%%%%%
\subsection{Projection Domains and Influence Operations}
%%%%%%%%%%%%%%%%%%%%%%%%%%%%%%%%%%%%%%%%%%%%%%%%%%%%%%%%%%%%%%%%%%%%%%%%

The influence set framework captures cascading dependencies and forms the foundation
for our graph-theoretic analysis. To analyze the propagation of influence
in the auction network, we construct influence sets using a sequence of 
projection operations on the underlying bid graph. Each active bid is indexed by 
a pair $(i, j) \in \mcB \times \mcL$, where buyer~$i$ submits a bid to seller~$j$. 
These interactions collectively form the strategy space $S(t)$, which consists of 
the full collection of price-quantity bids $s_i^j$. We extract the active subgame
by identifying $\mcI_{\text{active}}(t) \subseteq \mcB \times \mcL$, the set of
observed interactions between buyers and sellers at time~$t$. These pairs serve 
as both an interaction graph and an index set for the time-dependent strategy
array $\mathbf{s}(t)$.

%–––––––––––––––––––––––––––––––––––––––––––––––––––––––––
\subsection{Projection-Based Influence Propagation}

Let $\mcI_{\text{active}}(t)\subseteq\mcB\times\mcL$ be the set of buyer–seller pairs 
that submit positive bids at time~$t$. Each pair $(i,j)$ indexes the strategy 
array $\mathbf s(t)\in S(t)$, so $\mcI_{\text{active}}(t)$ is both a graph on
$\mcB\cup\mcL$ and an index set for the strategy space.

\noindent\textbf{Projection maps.}
\[
\pi:\mcI_{\text{active}}(t) \longrightarrow \mcB, \qquad \pi(i,j)=i, \qquad
\varpi:\mcI_{\text{active}}(t) \longrightarrow \mcL, \qquad \varpi(i,j)=j.
\]

\begin{itemize}
    \item \emph{Structural role.}  Alternating compositions
          $\varpi \circ \pi^{-1} \circ \pi \circ \varpi^{-1}\dots$
          trace paths through the bipartite auction graph, giving
          $n$-hop neighborhoods.
    \item \emph{Strategic role.}  Acting on $S(t)$, the same maps carve out
          partial strategy profiles (e.g. all bids of a given buyer).
\end{itemize}

\noindent\textbf{Full pre-images.}  
We take the composition of the the projections in order to restrict and 
vectorize the space $S(t)$. For any buyer~$i$ and seller~$j$ we write
\[
\varpi^{-1}(i,t)=\{(i,j')\in\mcI_{\text{active}}(t)\},
\quad
\pi^{-1}(j,t)=\{(i',j)\in\mcI_{\text{active}}(t)\},
\]
so that
\[
\varpi \circ \pi^{-1}(i,t) = \{j'\mid(i,j')\in\mcI_{\text{active}}(t)\} = \mcL_i(t),
\quad
\pi \circ \varpi^{-1}(j,t) = \{i'\mid(i',j)\in\mcI_{\text{active}}(t)\} = \mcB^{j}(t).
\]

\noindent\textbf{$n$-step influence.}  
Because the graph is bipartite, two successive projections always return a vertex 
of the \emph{same} type. These projections serve dual purposes: structurally, they 
trace paths through the auction graph, alternating between buyers and sellers; 
strategically, they extract subspaces of $S(t)$ that represent partial strategies
or responses and may evolve as patterns in the form of active bids sets.

%–––––––––––––––––––––––––––––––––––––––––––––––––––––––––
\paragraph*{Connected components via iterated projections}

Because the auction graph is bipartite, two successive projections return a vertex of the 
same type.  Define the composition operators
\[
P:=\varpi\circ\pi^{-1}\quad(\text{buyer}\to\text{seller}),\qquad
Q:=\pi\circ\varpi^{-1}\quad(\text{seller}\to\text{buyer}).
\]

\noindent
Starting from a seller~$j$, one step of influence expansion is
\[
\Lambda_{\mcL}^{(1)}(j,t)=P\ Q\ (j)=
\varpi\circ\pi^{-1}\circ\pi\circ\varpi^{-1}(j,t),
\]
which moves $\text{seller}\to\text{buyers}\to\text{sellers}$.
Analogously, for a buyer~$i$ we set $\Lambda_{\mcB}^{(1)}(i,t)=Q P(i)$.

\noindent
The $n$-hop neighborhoods follow by simple iteration,
\[
\Lambda_{\mcL}^{(n)}(j,t)=(P\ Q)^{\ n}(j),\qquad
\Lambda_{\mcB}^{(n)}(i,t)=(Q\ P)^{\ n}(i),\qquad n\ge1,
\]
with the base case $\Lambda^{(0)}(x,t)=\{x\}$.

Each application of $PQ$ (or $QP$) adds exactly one buyer–seller alternation, 
so $\Lambda^{(n)}(x,t)$ is the breadth-first shell lying $n$ hops away from $x$.  
Iterating until $\Lambda^{(n)}(x,t)=\Lambda^{(n-1)}(x,t)$ closes the connected 
component containing $x$.  Thus the recursive projection operator captures both 
direct and indirect influence flows. Thus, $\Lambda^{(n)}$ may be interpreted as 
the $n$-step neighborhood in the auction graph and as a dynamic closure of 
best-response behavior. Each expansion layer captures not just structural 
proximity but strategic influence--the transmission of incentive, information, 
and utility across the network. In this way, we convert local participation 
patterns into global influence propagation, formalized as graph-theoretic 
expansions over the projected structure of the strategy space.

%–––––––––––––––––––––––––––––––––––––––––––––––––––––––––
\subsection{Partial Ordering and Market Shifts}

While the projection mappings $\pi$, $\varpi$ and their compositions produce index sets
(subsets of $\mcB$ or $\mcL$), these sets are characterized by the underlying
strategy space $S(t)$. These indices correspond directly to elements of the strategy
space $S(t)$, which contains structured bid information $s_i^j(t)$. That is, the projection 
$\varpi^{-1}(i,t)$ retrieves all bid tuples $(i, j)$ in the index set, but equivalently
defines the subspace of $s_i(t)$ consisting of all bid array submitted by buyer~$i$ at time~$t$. 

Each element $s^j(t) \in S^j(t)$ is a bid array, and the collection $\pi \circ \varpi^{-1}(j,t)$.
Aside from being an index set of buyers, we have a set of bid arrays 
$\{s_i^j(t)\}_{i\in\mcB^j(t)}$ that can be 
partially ordered by their prices $p_{(n)}^j(t)$. While the projection operators isolate 
buyers or sellers structurally, the \emph{functional influence} between market participants 
is mediated through the comparison of bid prices. 

This introduces a natural partial order
among bidders for each seller at a fixed time, and we define a partial ordering on $S^j(t) \subset S(t)$
by $p_i^j(t) < p_k^j(t)$ if buyer~$i$ bids less than buyer $k$ for the same seller~$j$.
Given the set of buyers $\mcB^j(t) = \pi \circ \varpi^{-1}(j,t)$, we may impose a partial 
order structure based on the associated price bids $\{p_i^j(t)\}$. This ordering determines 
which buyers are accepted by seller~$j$ (those with highest prices until the resource is 
exhausted), and the sensitivity of a potential equilibrium to shifts in bidding behavior.

As shown in~\cite{Blocher2021}, a market shift occurs precisely when
a buyer outside of $\mcB^j(t)$ improves their relative position in this order, causing $\mcB^j(t)$ 
to be recomputed. Such shifts reflect structural changes in functional influence, 
as they alter the competitive hierarchy among bids and propagate through the
multi-auction environment.
Specifically, a market shift in auction~$j$ occurs when the partial order of bids at 
seller~$j$ changes in a way that affects allocation.
From \cite{Blocher2021}, two cases are critical:
%–––––––––––––––––––––––––––––––––––––––––––––––––––––––––
\begin{enumerate}
    \item \textbf{Demand Shortfall:} A buyer $i \in \mcB^j(t)$ reduces their bid quantity so
    that total demand falls below available supply:
    \[
    \sum_{i\in\mcB^j(t)} a_i^j(t) < Q^j(t).
    \]
    The auction must recompute its reserve price or reallocate supply among remaining buyers.      
    \item \textbf{Bid Overtake:} A buyer $i^* \notin \mcB^j(t)$ improves their valuation so that
    \[
    p_{i^*}^j(t) < p_*^j(t),
    \]
    where $p_*^j(t)$ is the minimum accepted bid price at auction~$j$. The buyer $i^*$ 
    displaces the marginal buyer, triggering a shift in $\mcB^j(t)$.
\end{enumerate}
%–––––––––––––––––––––––––––––––––––––––––––––––––––––––––
\noindent
Either case changes the minimal winning price, breaking the partial ordering within the
projected sets, and forces the auction to recompute $\mcB^{j}(t)$. As the seller frontier 
$\Lambda^{(1)}_{\mcL}(j,t)$ consists of sellers connected to~$j$ through shared buyers, 
the reallocation thereby propagates influence through layers of the the expansion 
$\Lambda^{(n)}_{\mcL}(j,t)=(\varpi\circ\pi^{-1}\circ\pi\circ\varpi^{-1})^{n}(j)$. 

%================= PREAMBLE =================
% \usepackage{pgfplots}
% \pgfplotsset{compat=1.18}
% \usepackage{tikz}
% \usetikzlibrary{arrows.meta,positioning}
%===========================================

\begin{figure}[t]
\centering
\begin{tikzpicture}[
  arr/.style={-{Latex[length=3mm]}, line width=0.6pt}
]

% -------------------- PARAMETERS --------------------
% Buyers (rows) I, Sellers (cols) J
\def\I{6}
\def\J{5}
% View and sizes
\def\viewAz{100}   % azimuth (turns around z-axis)
\def\viewEl{55}   % elevation (lower → more horizontal surface)
\def\w{6.9cm}     % previous 7.0
\def\h{5.2cm}     % previous 5.0
%\def\w{7.0cm}     % axis width
%\def\h{5.0cm}     % axis height

% -------------------- STATE A --------------------
\node[anchor=south west] at (0,-1) {\textbf{State A (before)}};

\begin{axis}[
  at={(0cm,0.4cm)}, anchor=south west,
  width=\w, height=\h, view={\viewAz}{\viewEl},
  xlabel={seller $j$}, ylabel={buyer $i$}, zlabel={tier},
  xmin=0.5, xmax=\J+0.5, ymin=0.5, ymax=\I+0.5, zmin=0, zmax=3,
  xtick={1,...,\J}, ytick={1,...,\I}, ztick={0,1,2,3},
  axis on top,
  colormap/viridis,
  zmajorgrids, ymajorgrids, xmajorgrids,
  grid style={black!15},
  enlargelimits=false, clip=true,
  domain=1:\J, y domain=1:\I
]

% --- (A) Step surface for tiers per seller ---
% We encode tier as: 3 = high band, 2 = mid, 1 = low, 0 = below grid.
% Here: same split for all sellers: rows 1-2 -> 3, 3-4 -> 2, 5-6 -> 1.
\addplot3[surf, opacity=0.85, shader=interp, samples=10, samples y=10]
  ({x},{y},
   { (y<=2) * 3 + (y>2 && y<=4) * 2 + (y>4) * 1 });

% --- (A) Active bids (dots) and marginal (ring) ---
% Format: (j,i,tier_offset) where tier_offset slightly lifts the dot above surface
\addplot3+[only marks, mark=*, mark size=1.2pt] table[row sep=\\] {
% j   i   z
  1   3   1.02\\ % (i3,j1)
  1   1   3.02\\ % (i1,j1)
  2   5   1.02\\ % (i5,j2)
  2   2   3.02\\ % (i2,j2)
  2   3   2.02\\ % (i3,j2)
  3   6   1.02\\ % (i6,j3)
  3   4   2.02\\ % (i4,j3)
  4   1   3.02\\ % (i1,j4)
  4   2   3.02\\ % (i2,j4)
  5   5   1.02\\ % (i5,j5)
  5   4   2.02\\ % (i4,j5)
};

% Marginal winners as open circles (same coords as marginal cells)
\addplot3+[only marks, mark=o, mark size=3.6pt, line width=0.9pt] table[row sep=\\] {
  1 2 3.02\\ % (i2,j1) marginal at high
  2 3 2.02\\ % (i3,j2) marginal at mid
  3 4 2.02\\ % (i4,j3) marginal at mid
  4 2 3.02\\ % (i2,j4) marginal at high
  5 4 2.02\\ % (i4,j5) marginal at mid
};

\end{axis}

% -------------------- Arrow --------------------
\node (sepA) at (\w-0.8cm,3.9cm) {};
\node (sepB) at (\w+0.2cm,3.9cm) {};
\draw[arr] (sepA) -- (sepB) node[midway, above] {\small phase transitions};

% -------------------- STATE B --------------------
\node[anchor=south west] at (\w+3.2cm,-1) {\textbf{State B (after)}};

\begin{axis}[
  at={(\w+0.5cm,0.4cm)},   anchor=south west,
  %at={(\w+3.2cm,0.4cm)}, anchor=south west,
  width=\w, height=\h, view={\viewAz}{\viewEl},
  xlabel={seller $j$}, ylabel={buyer $i$}, zlabel={tier},
  xmin=0.5, xmax=\J+0.5, ymin=0.5, ymax=\I+0.5, zmin=0, zmax=3,
  xtick={1,...,\J}, ytick={1,...,\I}, ztick={0,1,2,3},
  axis on top,
  colormap/viridis,
  zmajorgrids, ymajorgrids, xmajorgrids,
  grid style={black!15},
  enlargelimits=false, clip=true
]

% --- (B) Step surface after transitions ---
% Example transitions:
% * j=1 gets a new higher tier participant (i6 enters high),
% * j=1 loses a low cell (demand shortfall),
% * tiers remain the same (we keep the same banding, but points change).
\addplot3[surf, opacity=0.85, shader=interp, samples=10, samples y=10]
  ({x},{y},
   { (y<=2) * 3 + (y>2 && y<=4) * 2 + (y>4) * 1 });

% --- (B) Active bids updated ---
\addplot3+[only marks, mark=*, mark size=1.2pt] table[row sep=\\] {
  1  6  1.02\\ % NEW high participant enters stack (visual shift by point)
  1  2  3.02\\
  1  1  3.02\\
  2  5  1.02\\
  2  1  3.02\\ % reconfiguration at j2
  2  2  3.02\\
  3  6  1.02\\
  3  4  2.02\\
  4  3  2.02\\ % small shuffle at j4
  4  1  3.02\\
  5  5  1.02\\
  5  4  2.02\\
};

% Marginal winners after
\addplot3+[only marks, mark=o, mark size=3.6pt, line width=0.9pt] table[row sep=\\] {
  1  1  3.02\\ % (i1,j1) now marginal
  2  2  3.02\\
  3  4  2.02\\
  4  1  3.02\\
  5  4  2.02\\
};

% --- Optional: translucent “tier shift” slab to highlight a change on z ---
% Example: highlight a high-tier emphasis at seller j=1
\addplot3[
  fill=black, fill opacity=0.08, draw=black!30, line width=0.2pt
]
coordinates {
  (0.85,0.5,3) (1.15,0.5,3) (1.15,6.5,3) (0.85,6.5,3)
} -- cycle;

\end{axis}

\end{tikzpicture}
\caption{3D matrix view: rows (buyers), columns (sellers), and $z$ encodes price tiers. The colored surface shows the buyer price; filled markers are active bids; open circles show marginal winners. The right panel shows a transition where a new high-tier participant appears at $j{=}1$, a demand shortfall removes a low cell, and a reconfiguration shifts activity at $j{=}2$.}
\end{figure}
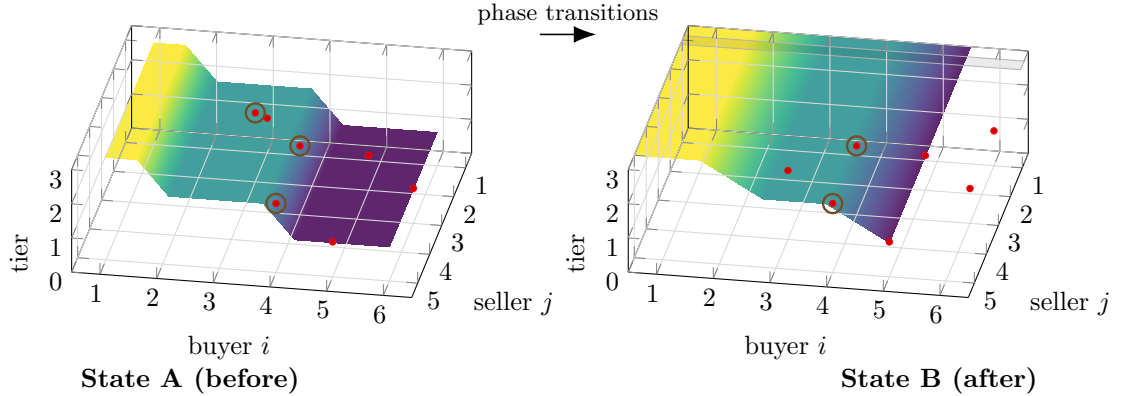

%–––––––––––––––––––––––––––––––––––––––––––––––––––––––––
\section{Influence Shells and Local Saturation}\label{sec:saturation}

The projection and ordering framework developed above allows us to describe how influence 
propagates through the auction network. We now ask under what conditions this propagation 
stabilizes. As buyers and sellers iteratively adjust bids, certain neighborhoods of the 
market reach a state in which no participant can improve their utility through unilateral
deviation. These locally stable regions form influence shells--bounded subgraphs within 
which allocations, prices, and bid updates remain consistent under further iterations.
When every buyer and seller in such a region satisfies this best-response property, 
the shell is said to be saturated. The following gives the formal notation and 
enumerates the assumptions that we have made in the generalization of influence sets as
were defined in~\cite{Blocher2021}.

%–––––––––––––––––––––––––––––––––––––––––––––––––––––––––
\begin{definition}[Saturated Influence Shell]\label{def:saturated_shell}
A primary influence set $\Lambda^{(1)}_{\mcL}(j,t)$ associated with seller~$j$ at time~$t$ 
is said to be saturated if no buyer or seller within this set can improve their 
utility by unilaterally altering their bids. Formally, for every buyer $i \in \mcB^j(t)$ 
and every seller $\ell \in \Lambda^{(1)}_{\mcL}(j,t)$, the following holds,
\[
u_i(t) \geq u_i'(t), \quad \text{for any feasible alternative strategy} \quad s_i'(t).
\]
\end{definition}
%–––––––––––––––––––––––––––––––––––––––––––––––––––––––––
\noindent
Global market equilibrium decomposes into interconnected saturated shells, each functioning as 
stable subsystems. We establish conditions for the existence of a saturated influence shell.

\begin{enumerate}[(i.)]
\item \textbf{Countable and Locally Finite Graph.}
The sets of buyers $\mcB$ and sellers $\mcL$ are at most countably infinite. 
Each participant engages in only finitely many transactions, ensuring finite degree at 
every instant. This guarantees that all projection maps ($\pi,\varpi$) encounter only 
finite fibres and that the influence operator \eqref{eq:Lambda-recursive} perform finite unions.

Buyers and sellers participate in locally finite networks, enabling stable equilibrium 
convergence within compact, bounded strategy spaces. Market rules explicitly 
limit resources and interactions, ensuring finite dimensionality.

\item \textbf{Bounded Influence and Bids.}
Influence propagation strength remains bounded, preventing divergence. Each seller’s 
fixed endowment $Q^j$ and each buyer’s fixed demand cap $Q_i$ are finite and 
time-independent. Hence every non-zero bid quantity~$q_i^j(t)$ lies in the 
compact interval $[0,Q_i]$, and every realized allocation is in $[0,Q^{j}]$.

\item \textbf{Partial Ordering Stability.} The partial ordering induced by the 
bid structure must satisfy stable bid threshold rankings for all relevant buyers 
and sellers within the shell, thus establishing clear marginal price tiers.
\end{enumerate}

\subsection{Ordering and Influence Propagation}

We recall the ordering relationship from~\cite{Blocher2021} that holds for any seller 
within a saturated primary influence set $\Sigma := \Lambda^{(1)}_\mcL(j,t)$. 
%------------------------------------------------------------
\begin{lemma}[Local Price Ladder {\cite[Thm.~2.3, proof]{Blocher2021}}]\label{lem:ladder}
Let the market be at time $t$ with seller $j$ and its saturated primary
influence set
$\Sigma := \Lambda^{(1)}_{\mcL}(j,t)$.  Pick any neighbor
$k\in\Sigma$ and two buyers
\[
  i \in \mcB^{j}(t), \qquad
  \ell \in \Lambda_{\mcB}^{(1)}(i,t)\setminus\mcB^{j}(t),
\]
i.e., $i$ bids on $j$, $k$ is another seller reached from $i$, and
$\ell$ is a buyer that bridges further to $k$ but not to $j$.

If the shell 
$\Lambda^{(1)}_{\mcL}(j,t)$ is saturated--no profitable deviation exists for any 
vertex in this set--because $\ell$ does not bid on~$j$ while~$i$ bids on both~$j$ 
and~$k$, the marginal thresholds must nest as follows
\begin{equation}\label{eq:interval-ordering}
  \sup \big(\overline p^k(t),\, \underline p^k(t)\big)\ \le\ p_\ell^{*}(t)
  \ < \inf \big(\overline p^j(t),\, \underline p^j(t)\big)  \ \le\ p_i^{*}(t) ,
\end{equation}
where the marginal intervals are chosen as in~\eqref{eq:margins}, 
and in particular, every choice of reserve prices
$p_*^k(t)\in \big(\overline p^k(t),\, \underline p^k(t)\big)$
and $p_*^j(t)\in \big(\overline p^j(t),\, \underline p^j(t)\big)$ satisfies
\begin{equation}\label{eq:ordering}
  p_*^{k}(t)\ \le\ p_\ell^{*}(t)\ <\ p_*^{j}(t)\ \le\ p_i^{*}(t).
\end{equation}
\end{lemma}

\begin{proof}[Proof]
The argument for~\eqref{eq:ordering}
is contained within the proof of \cite[Thm.~2.3]{Blocher2021}. Left–to–right the chain reads
\begin{enumerate}[(i.)]
  \item $p_*^{k}$ seller~$k$’s reserve price;
  \item $p_{\ell}^{*}$ buyer $\ell$’s bid that clears the marginal tier in \emph{both} auctions;
  \item $p_*^{j}$ seller~$j$’s reserve price;
  \item $p_{i}^{*}$ the highest active bid of buyer~$i$ on seller~$j$.
\end{enumerate}

The left inequality holds because~$k$ clears at the minimum of the two buyers’ bids; 
the strict middle inequality follows from~$\ell$ placing no bid on~$j$; the right inequality
is enforced by buyer~$i$’s cross-auction participation.
In a saturated shell, any profitable deviation by $k$ toward $j$ or by
$i$ away from $j$ is ruled out.  Hence buyer $\ell$'s marginal valuation
must lie weakly above all prices at which $newk$ can clear its final
unit, while buyer $i$'s valuation must lie weakly below all prices at
which $j$ can still clear.  This forces the margin intervals to nest
as in~\eqref{eq:interval-ordering}.  Since
$p_*^k(t)\in \big(\overline p^k(t),\, \underline p^k(t)\big)$
and $p_*^j(t)\in \big(\overline p^j(t),\, \underline p^j(t)\big)$
by the definition of the reserve price $p_*^j$,
the pointwise ladder~\eqref{eq:ordering} follows immediately.
\end{proof}

Economically, this implies that higher prices at neighboring sellers prevent 
buyers from deviating profitably, ensuring that no participant has an incentive to alter
their bidding strategy unilaterally. Hence, the ordering captures a stable distribution of 
resources and prices, reflecting locally optimal market conditions.

\paragraph*{Weak (local) Monotonicity}

The projected influence sets, together with the induced partial orders, thus form the
dynamic framework for market evolution, where projections identify which vertices 
are connected, and partial orders determine how influence is transmitted via price 
shifts. Market shifts occur when the partial order structure is perturbed beyond
certain thresholds, forcing recomputation of $\mcB^j(t)$ or reserve prices.

%------------------------------------------------------------
\begin{proposition}[Local Monotonicity]\label{prop:local-monotone}
Let $\Sigma = \Lambda^{(1)}_{\mcL}(j,t)$ be the saturated one-hop shell
of seller~$j$ at time~$t$. 
For each seller $k\in\Sigma$, let the marginal intervals be chosen as in~\eqref{eq:margins},
and let the reserve price $p_*^k(t)$ be selected
inside the interval $(\overline p^k(t),\, \underline p^k(t))$.

Now, consider the vector of seller reserves and buyer marginal valuations
restricted to $\Sigma$. 
Under elastic, strictly decreasing valuation functions $\theta_i'(z)$
and any reserve-update rule that
\begin{enumerate}[(i)]
\smallskip
  \item selects $p_*^k(t+1) \in (\overline p^k(t+1), \underline p^k(t+1))$ for each $k\in\Sigma$, and
  \smallskip
  \item is nondecreasing in the bids $\{p_i^k(t)\}_{i}$ and the previous
        reserve $p_*^k(t)$,
\smallskip
\end{enumerate}
the PSP price–update map
\begin{equation}\label{eq:weak-monotone}
   p_*^k(t)  \mapsto \
   \tilde p_*^k(t+1)\quad \text{and} \quad p_i^k(t) \mapsto \tilde p_i^k(t+1)
\end{equation}
is locally monotone on $\Sigma$.
\end{proposition}

%------------------------------------------------------------
\begin{proof}
Elasticity of the valuation functions implies that each buyer’s marginal
value $p_i^{*}(t)=\theta_i'(z_i(t))$ is strictly decreasing in its own
allocation. Increasing any bid or reserve inside the saturated shell
$\Sigma$ may raise some allocations and lower others, but it cannot
create a reversal of the bid ordering that defines the interval
$(\overline p^k(t),\, \underline p^k(t))$. In particular, all ladder
relations \eqref{eq:ordering} remain invariant under such updates.

For sellers, the reserve-update rule is assumed nondecreasing in both
the previous reserve $p_*^k(t)$ and the local bids $\{p_i^k(t)\}_{i}$.
Thus any componentwise increase in $(p_*^k(t),p_i^k(t))$ cannot decrease
any updated reserve $p_*^k(t+1)$. By construction, each updated reserve
remains inside its interval $(\overline p^k(t+1),\,\underline p^k(t+1))$,
and, because the shell $\Sigma$ is saturated, these intervals evolve
compatibly with the ordering relations and cannot induce a downward jump
that violates the ladder.

Therefore, every coordinate of the updated pair
$(p_*^k(t+1),\,p_i^{k}(t+1))$ is weakly increasing in the corresponding
coordinate of $(p_*^k(t),p_i^k(t))$. Hence the update rule satisfies the
order-preserving property \eqref{eq:weak-monotone}, and the PSP
price–update map is locally monotone on $\Sigma$.
\end{proof}

%–––––––––––––––––––––––––––––––––––––––––––––––––––––––––

The partial ordering structure induced by bidding behavior is essential in 
analyzing and predicting the direction and magnitude of market shifts resulting from 
influence dynamics. A local allocation triggers global bid adaptation, reinforcing 
that while seller auctions operate independently, buyer strategy space remains tightly 
coupled. In integrated markets (scarce supply), the partial orders are dense and
tightly coupled, making markets highly sensitive and globally coordinated. 
In fragmented markets (abundant supply), the partial orders become sparse and disconnected, 
leading to localized equilibria and insulating submarkets from external shocks. 
Thus, the transition from integrated to fragmented equilibrium is not just a graph 
phenomenon--it is a transition in the connectivity of the partial order structure induced by bidding.

\begin{remark}[On the ordering of the PSP price map]\label{rem:map-ordering}
Although the local update map is written on the pair $(p^{j},p_{i}^{j})$,
the PSP rule actually \emph{acts} only on the first coordinate:
the buyer’s bid $p_{i}^{j}$ is simply carried forward 
(or set to $0$ if $j\notin\mcL_i$).
\end{remark}

%–––––––––––––––––––––––––––––––––––––––––––––––––––––––––
\subsection{Asynchronous Sellers and Coupled Buyers}
%–––––––––––––––––––––––––––––––––––––––––––––––––––––––––

To represent the fine-grained dynamics of bid selection and displacement within an auction
round, we introduce an internal index~$\tau_k$ to describe local progression steps.  Each
$\tau_k$ denotes a partial–ordering resolution event—an allocation decision at auction~$j$
followed by an update to the reserve price and potentially to the projected sets.  The index
acts as a local time variable inside the global iteration~$t$, allowing us to separate
micro–adjustments from round-to-round evolution.

%–––––––––––––––––––––––––––––––––––––––––––––––––––––––––
\begin{definition}[Allocation Step $\tau_k$] 
At each $\tau_k$ within round~$t$, seller~$j$ selects the highest bidder
in $\pi\circ\varpi^{-1}(j)$ not yet fulfilled, allocates a feasible amount
$a_i^j(\tau_k)$, updates the reserve price $p_*^j(\tau_{k+1})$, and recomputes $\mcB^j$ and
$\Lambda_\mcL(j)$ as needed. 
\end{definition}
%–––––––––––––––––––––––––––––––––––––––––––––––––––––––––

Each $\tau_k$ inside a global round~$t$ is therefore a local ordering–resolution event: seller~$j$
picks the highest unfilled bidder, allocates a feasible amount $a_i^{j}(\tau_k)$, updates its
reserve, and recomputes the bidder set~$\mcB^{j}(t)$ as well as the seller shell
$\Lambda^{(1)}_{\mcL}(j,t)$.  The sequence $\tau_1,\tau_2,\dots$ terminates when no remaining
buyer meets the current reserve or when supply is exhausted.

Sellers operate independently: each seller’s $\tau_k$ sequence proceeds without synchronization
with others.  However, buyers must maintain a consistent strategy across all sellers they bid on.
Since the buyer’s bid array~$\sigma_i(t)$ is defined jointly over~$\mcL_i(t)$, any change to the
outcome of one auction requires coordinated updates across all components.  This coupling between
independent seller threads through shared buyers produces the feedback mechanism responsible for
market coherence.  From a game–theoretic perspective, each seller executes a local best–response
process, while each buyer enforces cross–auction consistency of marginal valuation.

We define the buyer update rule as
\begin{equation}\label{eq:buyer_update}
    q_i^{j}(\tau_{k+1}) = Q_i(t) -
    \sum_{j' \in \mcL_i(t)} a_i^{j'}(\tau_k),
    \qquad \forall j \in \mcL_i(t),
\end{equation}
indicating that the buyer updates all bids simultaneously based on observed allocations.
Equation~\eqref{eq:buyer_update} ensures that a buyer’s total requested quantity never exceeds its
available resource~$Q_i(t)$ and redistributes residual demand across the active seller set~$\mcL_i(t)$.  The
rule formalizes how buyers translate local allocation feedback into revised offers, maintaining a
form of budget balance across asynchronous auctions.

Because sellers run their $\tau_k$ threads asynchronously while buyers must update all bids
coherently according to~\eqref{eq:buyer_update}, local price changes propagate through the partial
orders defined above, layer by layer.  Each seller’s reserve adjustment initiates a chain of
bid updates in the neighborhoods that share its buyers.  This extends the projection–ordering
framework, enabling intra–round modeling of bid dynamics, bid–induced reordering, and precise
tracking of influence propagation through updates to the projected domains.  In this sense, the
$\tau_k$ sequence acts as a micro–time resolution that reveals how local saturation unfolds inside
each global auction round.

\paragraph*{Stability under asynchronous evolution}
The following proposition shows that, even under these independent update threads, the relative
ordering of bids remains stable and the local price ladder
is preserved.
%–––––––––––––––––––––––––––––––––––––––––––––––––––––––––

\begin{proposition}[Saturated Shell]\label{prop:sat-shell}
Let $\Sigma=\Lambda^{(1)}_{\mcL}(j,t)$ be saturated at $\tau_k$. Assume that for every seller $k\in\Sigma$
the reserve price lies in the interval,
\[
    \overline p^k(\tau_k)\,<\, p_*^k(\tau_k) \,<\, \underline p^k(\tau_k),
\]
and that each $\tau$--update preserves all ladder relations~\eqref{eq:ordering} inside~$\Sigma$. 
If a local resolution step $\tau_k\to\tau_{k+1}$ modifies the strategy space only inside~$\Sigma$, then:
\begin{enumerate}[(i)]
\item The ladder~\eqref{eq:ordering} is preserved: no ordering reversal is possible, and every affected
marginal or reserve price weakly increases; a strict increase occurs whenever the winning bid or 
reserve at some seller in $\Sigma$ rises.
\item If every buyer or seller that first appears in $\Sigma_{n+1}\setminus\Sigma_n$ has no profitable 
deviation given the preserved ladder, then the expanded shell $\Sigma_{n+1}$ is saturated.
\end{enumerate}
\end{proposition}

%–––––––––––––––––––––––––––––––––––––––––––––––––––––––––

\begin{proof}
If a $\tau$--update reallocates quantity within $\Sigma$, the clearing price or reserve at the 
affected seller can only move upward within its margin interval. Because all reserves and winning 
bids in $\Sigma$ lie inside their intervals $(\overline p^k,\,\underline p^k)$, no update can
create a reversal of the ordering that defines the ladder~\eqref{eq:ordering}. Thus each affected
component moves weakly upward, and whenever the winning bid or reserve at some seller in $\Sigma$ 
increases, at least one of the four prices in the ladder strictly increases.

Saturation implies every $\tau$--update inside $\Sigma$ preserves best--response conditions given 
the ladder; hence extending $\Sigma$ by including agents in $\Sigma_{n+1}\setminus\Sigma_n$ yields
a saturated larger shell whenever those newly added agents also have no profitable deviation under 
the same ordering.
\end{proof}

Saturation implies every $\tau_k$ update inside $\Sigma$ is either a demand–shortfall or
bid–overtake event; both raise the marginal price they touch, propagating weakly upward along every
chain of the form \eqref{eq:ordering}.  A strict increase occurs whenever the winning bid or
reserve at some auction in $\Sigma$ is lifted.  Thus, local ``saturation'' is a best–response
property of a one–hop influence shell.

By an inductive test we extend saturation shell–by–shell.
\begin{corollary}\label{eq:propogated-shell}
The new shell $\Sigma_{n+1}$ inherits the price–ordering ladder \eqref{eq:ordering}: all its
marginal prices are no smaller than those in $\Sigma_n$; in particular,
    $p_*^{k}(\tau_{k+1})\ge p_*^{k}(\tau_k)$ for all $k\in\Sigma_{n+1}$.
\end{corollary}

\begin{proof}
Consider any buyer–seller quadruple $(k,\ell,j,i)$ whose seller $k$ lies in the freshly revealed
layer $\Sigma_{n+1} \setminus \Sigma_n$.  Applying
Proposition~\ref{prop:local-monotone} to that quadruple shows that the ladder
inequality~\eqref{eq:ordering} is preserved and all four prices weakly increase.  Repeating this
argument for every such quadruple that touches $\Sigma_{n+1}$ completes the extension of the
monotone ladder one hop outward.  For every edge $(i,k)$ with $k\in\Sigma_n$,
Proposition~\ref{prop:sat-shell} guarantees that a local $\tau$–update cannot decrease the marginal
price $p_*^{k}$.  Hence $\big(p_*^{k}\big)_{k\in\Sigma_n}$ is component–wise
non–decreasing from $\tau_k$ to $\tau_{k+1}$.
\end{proof}

The asynchronous update model developed above provides the conceptual foundation for our
simulation studies.  It captures the essential features of decentralized PSP dynamics: sellers
acting independently on local information, buyers coordinating across overlapping auctions, and
influence propagating through partially ordered interactions.  In the following section, we use
these principles to construct an event–driven simulation framework that allows us to observe
how local saturation emerges in practice.

\section{Simulation Framework and Implementation}
\label{sec:simulation-framework}

This section summarizes the simulation code used to study the PSP
markets with multiple sellers and buyers. 
The simulation architecture explores the practical 
realizations of decentralized coordination. The event-driven approach reproduces 
the iterative best-response behavior implied by the mechanism and allows 
examination of convergence properties, price dispersion, and efficiency loss 
due to network coupling. 

Following Semret and Lazar \cite{Lazarsemret1998}, each buyer's valuation is given by a
parabolic curve of the form
$$
    \theta_i(z)=\kappa_i (\bar q_i - z/2) z\qquad\hbox{for}\qquad
         z\in [0,\bar q_i]
$$
where $\bar q_i$ represents the maximum quantity of goods desired and
$\kappa_i=\bar p_i/\bar q_i$ has dimensions marginal price per unit where
$\bar p_i$ is the maximum marginal value that buyer $i$ would ever place
on the resource.

\subsection{Event-Driven Algorithm and Asynchronous Updates}
The simulation operates as a discrete-time event system. Events are scheduled 
and processed in a priority queue, advancing the simulation clock $t$ to
the next event. Two event types exist:
\begin{enumerate}
    \item \textbf{Buyer Compute:} Buyer $i$ evaluates its local 
    state, computes updated bids $(z_i^j, p_i^j)$ on each connected 
    seller $j$, and schedules bid events when meaningful changes occur.
    \item \textbf{Post Bid:} Seller $j$ clears its auction, applying 
    second-price allocation and updating quantities, payments, and revenues.
\end{enumerate}
Buyer and seller events may reschedule each other (e.g., clears triggered after
meaningful bid changes). The loop halts when no effective changes remain or a
step limit is reached. The simplified pseudocode is shown in Algorithm~\ref{alg:psp-sim}.

\begin{algorithm}[H]
\caption{Event-driven PSP simulation}
\label{alg:psp-sim}
\begin{algorithmic}[1]
\State Initialize market state $M$; schedule all buyers.
\While{queue not empty and not converged}
  \State $(t, \mathrm{type}, \mathrm{payload}) \gets \mathrm{pop}()$
  \If{$\mathrm{type} = \mathrm{BUYER\_COMPUTE}$}
    \State Update $(z_i^j, p_i^j)$ for buyer $i$ on feasible links.
    \State Schedule POST\_BID events for affected sellers.
  \ElsIf{$\mathrm{type} = \mathrm{POST\_BID}$}
    \State Seller $j$ clears auction, enforcing $Q^j$, opponent ordering, and payments.
  \EndIf
\EndWhile
\end{algorithmic}
\end{algorithm}

Each buyer $i$ computes a uniform (or per-seller) bid price using
a valuation-based update $w=\theta_i'(\sum_j z_i^j)$.
Quantities are apportioned across incident sellers using a
local best-response step.
Buyers are sorted by descending unit price $p_{(n)}^j$. 
The clearing process accumulates allocations until total demand equals available resource $Q^j$.
The threshold price
\begin{equation}
p_*^j = \min \{ p_i^j : \sum_{k: p_k^j \ge p_i^j} q_k^j \ge Q^j \}
\end{equation}
identifies the seller's marginal (clearing) price.

For each seller $j$, the clear routine builds a partial ordering
by posted marginal prices (bid prices), serves opponents until
available resource $Q^j$ is exhausted, and charges the price incurred by 
the externality of participation to all
served buyers for that seller. The routine updates \texttt{a}, 
seller revenue, and per-buyer costs. 
We define the set of active buyers with positive bids as
\begin{equation}\label{eq:active-bids}
\mathcal{I}_j = \{ i : q_i^j > 0, \ a_{ij} \in \mbfA \},
\end{equation}
where $A$ represents the biadjacency
matrix captures direct buyer–seller interactions:
\[
\mbfA(t) \in {0,1}^{|\mcB|\times |\mcL|},
\]
where $\mbfA_{ij}(t)=1$ if buyer~$i$ bids on seller~$j$ at time~$t$, 
and $\mbfA_{ij}(t)=0$ otherwise. Rows of $\mbfA(t)$ identify each buyer’s active sellers;
columns identify all buyers bidding on a given seller.

Experiments are conducted using randomized networks of $I=|\mcB|$ buyers and $J=|\mcL|$ sellers. 
The connectivity matrix $\mbfA_{ij}$ determines which buyers may interact with which sellers. 
Each run uses the following protocol:
\begin{enumerate}
    \item Initialize market state $M$ with parameters $(I, J, Q^j, \varepsilon, \mathrm{reserve})$.
    \item Assign buyer valuations $(\bar q_i, \kappa_i)$ and budgets $b_i$ from uniform ranges.
    \item Generate random biadjacency matrices with varying density (percentage of shared buyers).
    \item Execute the event-driven simulation for a fixed iteration limit.
    \item Record convergence statistics, prices, allocations, and revenues.
\end{enumerate}
The event scheduler supports both deterministic and stochastic updates, allowing controlled 
comparison between synchronous and asynchronous dynamics.

\paragraph*{Experimental Setup.}
Each experiment initializes a market with $I$ buyers and $J$ sellers.
Seller capacities are fixed at $Q_{\max} = [60.0, 40.0]$, with buyers
distributed across both sellers according to a connectivity percentage that
varies from $0\%$ (fully isolated) to $100\%$ (fully connected) in increments
of $10\%$. For each connectivity level, a base random seed
(\texttt{base\_seed = 20405008}) ensures reproducibility while allowing
controlled stochastic variation across runs.

Following Semret and Lazar \cite{Lazarsemret1998}, each buyer's valuation is given by a
parabolic curve of the form
$$
    \theta_i(z)=\kappa_i (\bar q_i - z/2) z\qquad\hbox{for}\qquad
         z\in [0,\bar q_i]
$$
where $\bar q_i$ represents the maximum quantity of goods desired and
$\kappa_i=\bar p_i/\bar q_i$ has dimensions marginal price per unit where
$\bar p_i$ is the maximum marginal value that buyer $i$ would ever place
on the resource.
Note that $\kappa_i$ is larger for buyers who derive more
value from the resource.  Now choose $\bar q_i$ and $\bar p_i$
independently for all $i$ such that
\begin{equation}\label{generousd}
    \bar q_i \sim U[50,100]\qquad\hbox{and}\qquad \bar p_i\sim U[10,20],
\end{equation}
where $U[a,b]$ represents the uniform distribution over the interval $[a,b]$.
Noise and perturbation effects are controlled by $\epsilon = 2.5$.
For each seed and
connectivity level, the simulation executes until convergence, measuring
clearing prices, allocations, and bid prices.

A sequence of derived seeds
$\texttt{seed} = \texttt{base\_seed} + s$ is used for each connectivity
level $s$, ensuring comparable random draws while preserving independence
across runs.

\subsection{Price Ladder Verification}

The simulation presented here focuses on verifying the \emph{price ladder condition} across 
interconnected sellers. This experiment represents a localized instance of the broader PSP
market, designed to test whether clearing prices obey a monotonic relationship when 
sellers share buyers through overlapping influence sets.

The experiment initializes a small market composed of two sellers ($j=1$ and $\ell=0$)
and four buyers ($i=0,1,2,3$). The adjacency structure allows some buyers to 
connect to both sellers, while others remain local. Sellers have distinct 
capacities, $Q^1=8$ and $Q^0=15$, reflecting asymmetric market sizes. 
The buyer valuation and bid initialization follow:
\begin{align*}
(0,1):& \; q=8,\; p=40, \qquad &(0,0):& \; q=8,\; p=40, \\ 
(1,0):& \; q=2,\; p=4, \qquad &(2,0):& \; q=6,\; p=1.
\end{align*}
We have the connectivity of the market,
\begin{figure}[h!]
    \centerline{\begin{minipage}[b]{0.75\textwidth}
    \caption{\label{fig:adj_lad}Adjacency structure showing market connectivity between buyers and sellers.}
    \end{minipage}}
    \centerline{\includegraphics[height=0.35\textwidth]{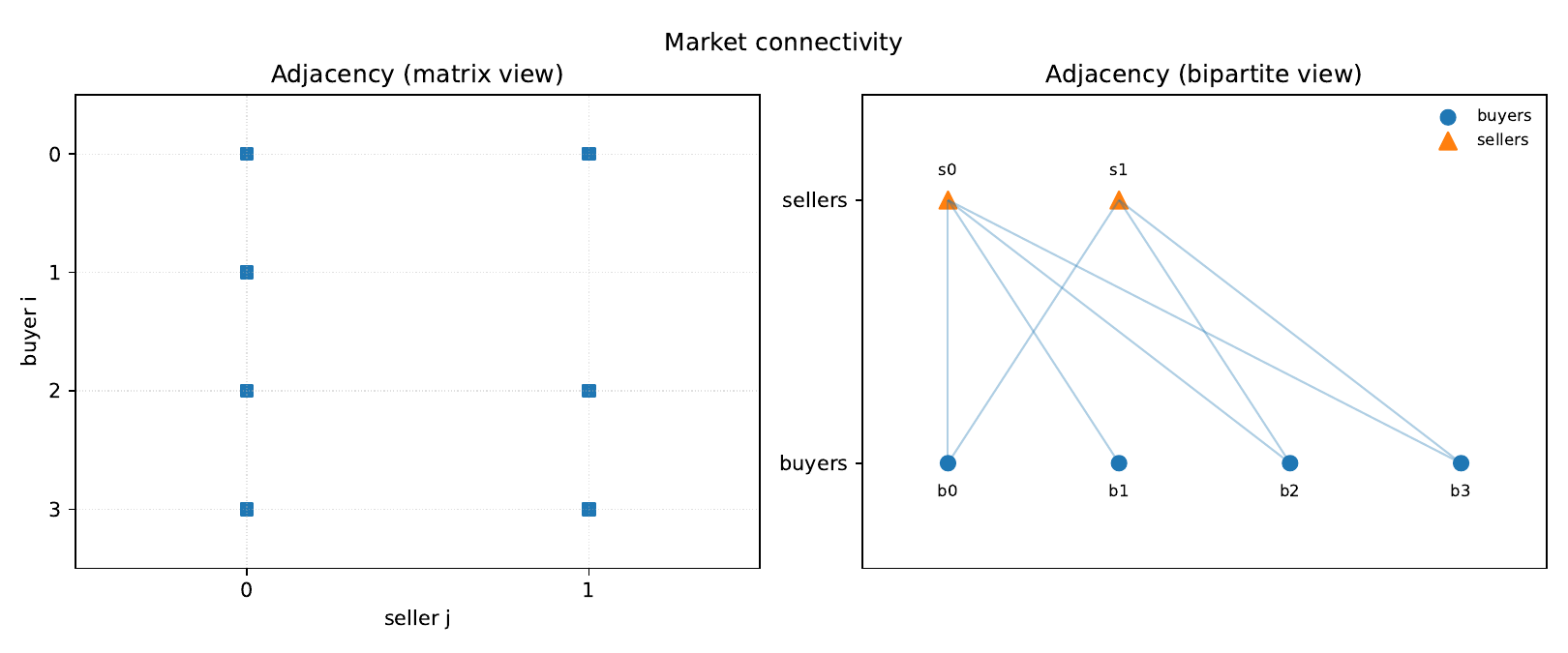}}
\end{figure}
in alignment with Lemma~\ref{lem:ladder}.

The resulting market has a 
highly skewed valuation distribution, allowing one buyer to dominate
both sellers, while others form the marginal tiers that define
second-price boundaries. 

The algorithm scans all sellers and their one-hop neighbors to 
evaluate tuples $(\ell, k, j, i)$ where Buyer~$i$ connects 
the two sellers. It tests the three inequalities defining the ladder ordering 
$p^*_\ell <  p_k < p^*_j \le p_i$.
If these inequalities hold for all tuples, the market
satisfies the monotone price ladder condition. Violations are 
reported with detailed tuple traces to aid in diagnosing market inconsistencies.

In this configuration, the ladder tuples satisfy all three 
inequalities, confirming a monotone relationship among clearing 
and bid prices. The system outputs a detailed report including,
number of valid tuples and unique seller pairs $(j,\ell)$,
margins between successive price tiers: 
  $(p_k - p^*_\ell)$, $(p^*_j - p_k)$, and $(p_i - p^*_j)$,
and a summary of any violations detected.
For this experiment, the output indicates no violations and consistent
monotonicity, demonstrating that the PSP clearing mechanism maintains
a globally ordered price structure when local competition and influence overlap exist.

This controlled experiment provides an analytical validation of the price 
ladder lemma in a simplified setting, and is intended to act as a 
unit test. It confirms that bid prices 
across connected sellers obey the expected inequalities implied 
Lemma~\ref{lem:ladder}. 
More generally, it shows that when buyers bridge multiple
sellers, the second-price mechanism induces a coherent ordering of
marginal prices, and provides an analytical tool for
extending this verification to larger graphs. 
\begin{figure}[h!]
    \centerline{\begin{minipage}[b]{0.75\textwidth}
    \caption{\label{fig:b0vcurvelad}Buyer~0 valuation curve and marginal diagnostics.}
    \end{minipage}}
    \centerline{\includegraphics[height=0.4\textwidth]{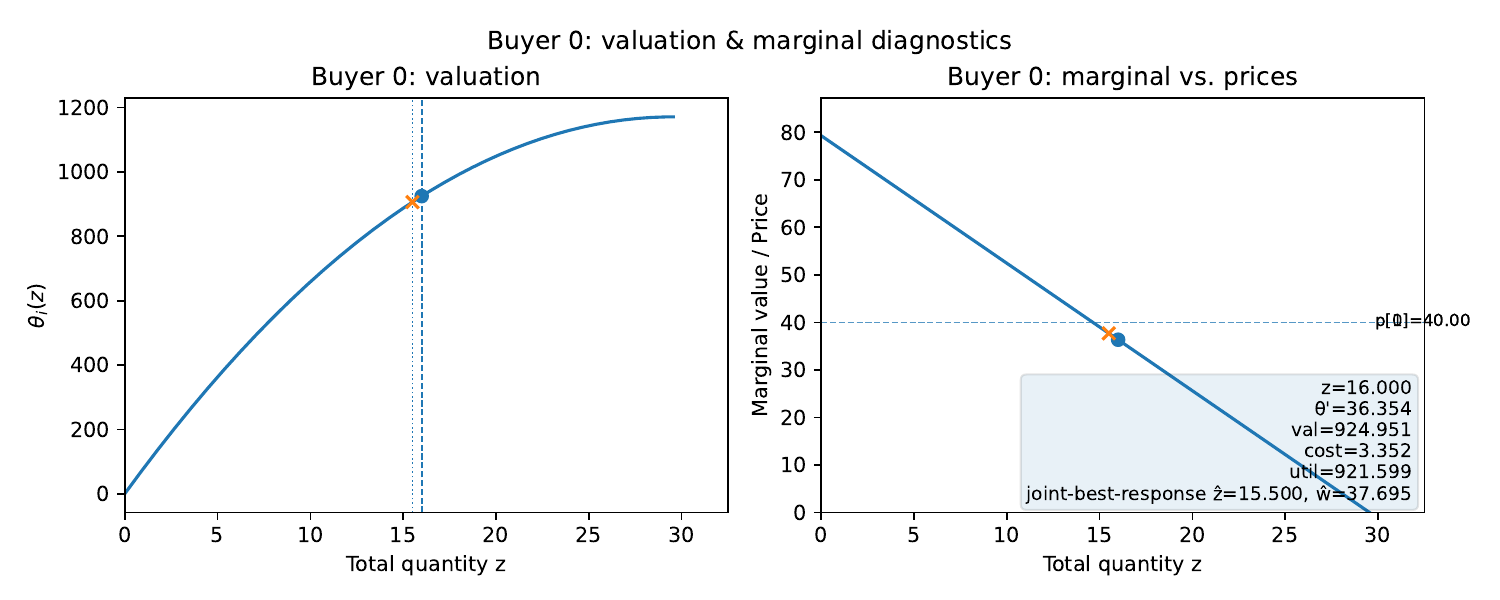}}
\end{figure}
In the tuples $(\ell,k,j,i)=(0,1,1,0)$ and $(0,2,1,0)$ we observe
\[
  p_\ell^*=1.0,\quad p_k\in\{4.0,1.0\},\quad p_j^*=40.0,\quad p_i=40.0.
\]
Thus the high-tier buyer at seller~$j$ sits \emph{at} the clearing price, while mid-tier competitors remain strictly below $p_j^*$. The reported margins
$(p_k-p_\ell^*)=0$, $(p_j^*-p_k)=36$, and $(p_i-p_j^*)=0$
reveal a wide central gap: a single dominant tier clears seller~$j$, whereas seller~$\ell$ is anchored by low-tier participation at a much smaller price.

The monotone relationship validated here provides empirical confirmation of Lemma~\ref{lem:ladder}. 
The experiment illustrates how buyers bridging sellers stabilize the market through consistent 
price ordering, even when capacities and bid magnitudes differ substantially. The asymmetry in seller
revenues and capacities demonstrates how equilibrium adapts to network structure, with high-valuation
buyers dominating smaller auctions and lower-tier participants anchoring larger ones.

Our next experiment allows us to observe the
propagation of equilibrium constraints across overlapping influence shells, 
offering empirical evidence for Propositions~\ref{prop:sat-shell} of this paper.

\subsection{Connectivity}

Further experimental results are aggregated as functions of the overlap percentage
between buyer--seller pairs, 
revealing how market interdependence affects stability, bid prices, 
and efficiency. The framework also enables sensitivity analysis under perturbations 
to parameters such as $\varepsilon$, budget distributions, and the structure of influence sets.

In this experiment, connectivity was gradually increased to observe 
how equilibrium formation and price alignment change as the market 
transitions from isolated to coupled seller networks. Starting 
from a sparse adjacency structure, buyers were allowed to 
participate in multiple auctions, creating overlaps that 
induced cross-seller influence and coupling of price dynamics.

Figure~\ref{fig:adj} illustrates the adjacency structure used
in the experiment. Connectivity defines the feasible market domain 
$I_{\mathrm{active}}(t)\subset \mcB\times \mcL$
that bounds all strategic interactions. 
\begin{figure}[h!]
    \centerline{\begin{minipage}[b]{0.75\textwidth} 
\caption{\label{fig:adj}Adjacency and market connectivity 
for the 8\,$\times$\,2 experiment. Connectivity is set at 50\%.}
    \end{minipage}}
    \centerline{
        \includegraphics[height=0.4\textwidth]{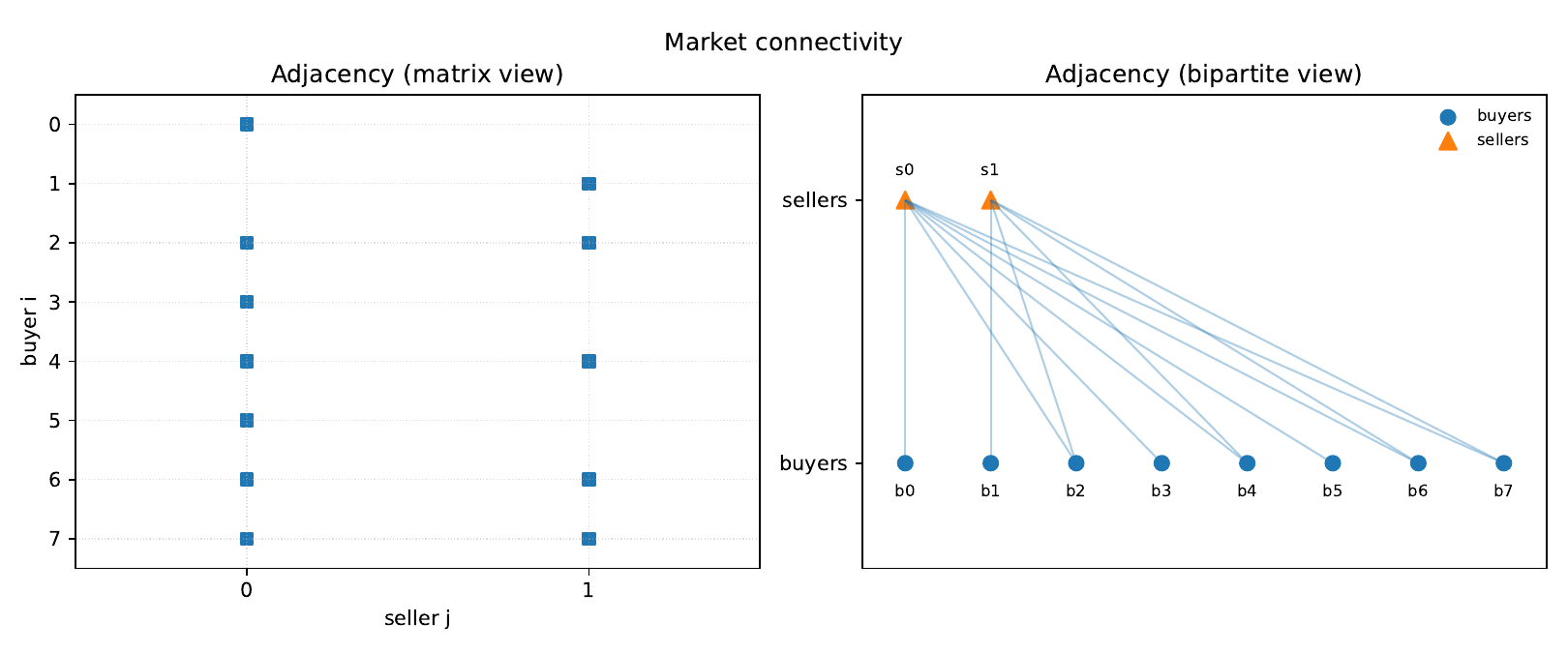}
    }
\end{figure}

Figure~\ref{fig:ub6s0} shows the utility surface for a single buyer–seller pair,
as was presented in~\cite{Lazarsemret1998} , 
here buyer $6$ and seller $0$, plotted over bid quantity $z_i^j$ and price $w_i^j$. 
The surface depicts the buyer’s instantaneous utility 
$u_i(z_i^j,w_i^j)=\theta_i(z_i^j)-z_i^jw_i^j$ given the opponent bids and current market
reserve. The concave ridge indicates the buyer’s optimal quantity at the current 
price level, while the lower regions show diminishing returns and cost-dominated outcomes.
We see a stable interior 
optimum: movements along the quantity axis correspond to allocation changes, whereas movements 
along the price axis reflect valuation gradients.
\begin{figure}[h!]
    \centerline{\begin{minipage}[b]{0.75\textwidth} 
\caption{\label{fig:ub6s0}%
Single buyer--seller utility surface for buyer~6 at seller~0. The surface plots
$u_i(z_i^j,w_i^j)=\theta_i(z_i^j)-z_i^j w_i^j$ over quantity $z_i^j$ and unit price $w_i^j$,
holding the opposing bids fixed at the snapshot.}
    \end{minipage}}
    \centerline{
        \includegraphics[height=0.5\textwidth]{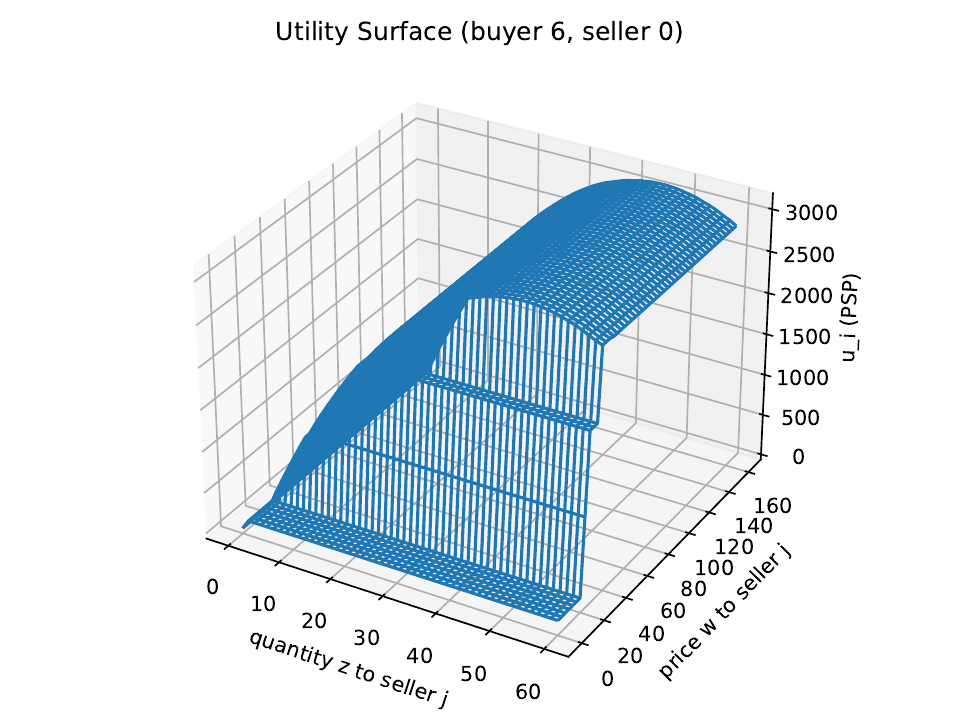}
    }
\end{figure}

%–––––––––––––––––––––––––––––––––––––––––––––––––––––––––––––––––––––
Next, we present an algorithm: an iterative evaluation of the
aggregate staircase $P_i(z,s_{-i})$.
At each iteration $t$, buyers and sellers perform the following operations:
\begin{algorithm}[ht]
\caption{Buyer Update Dynamics under Bounded Participation}
\begin{algorithmic}[1]
\State \textbf{Bid formation.} 
Each buyer $i$ applies the opt--out map 
$q_i(a(s):\mcL_i(t)) = [\ q_i^j(a)\ ]_{j\in\mcL_i(t)}$
and selects the minimal--cost subset of sellers.
\State \textbf{Utility evaluation \& Rebid.} 
Buyer $i$ computes the utility increment 
$\Delta u_i(t)$ from its updated bids.
Buyer $i$ updates its bids
iff $\Delta u_i(t)>\epsilon$.
\State \textbf{Allocation and clearing.}
Sellers allocate proportionally at each price $p_*^j(t)$.
Buyers at the cutoff price may receive partial allocations.
\State \textbf{Advance iteration.}
Set $t \gets t + 1$.
\end{algorithmic}
\end{algorithm}

\noindent
Because each accepted update increases some buyer’s utility by a bounded discrete
amount and the state space is finite, every sequence of threshold--improving updates
must terminate in an absorbing $\epsilon$--NE region.
We speculate that the induced dynamics
are weakly acyclic: from any initial state, at least one finite improvement path
leads to equilibrium.

Figure~\ref{fig:buyer6-diagnostics} are produced by the above construction. 
For Buyer~6 we evaluate the staircase $P_i(\cdot,s_{-i})$, compute $(q_i,w_i)$, 
perform the minimal–cost fill, and then read off the realized total $Z_i$ 
and price $p^*:=P_i(Z_i,s_{-i})$. The left panel shows
$(Z_i,\theta_i(Z_i))$ on the concave valuation curve; the right panel
shows $\theta_i'$ together with the dashed price level $p^*$. In the runs shown, 
the valuation is quadratic,
$$
  \theta_i(z) \approx a z - \tfrac{b}{2} z^2,\quad
  \theta_i'(z) = a - b z,\quad a\approx 66,\ b\approx 1.1,
$$
so marginal value declines approximately linearly from $\theta_i'(0)\approx 66$ 
to near zero around $z\approx 60$. Two sellers induce a two–step staircase in 
$P_i$; the three snapshots correspond to marginal price levels near $p^*\approx 32.1$ 
with $Z_i\approx 28.2$ (interior), $p^*\approx 28.6$ with $Z_i\approx 28.3$
(price–limited), and a high–availability case with $Z_i\approx 52.2$
where the buyer is constrained by feasibility at that price. These values 
are taken directly from the algorithm’s output and no post–hoc smoothing is applied.

% --- Two-panel figure, side-by-side: constraint-limited and price-limited cases ---
% Requires: \usepackage{graphicx} and \usepackage{subcaption}
\begin{figure*}[t]
  \centering

  % Panel A — Constraint-limited (θ'(ẑ_i) > p*)
  \begin{subfigure}{0.45\textwidth}
    \centering
    \includegraphics[width=\linewidth]{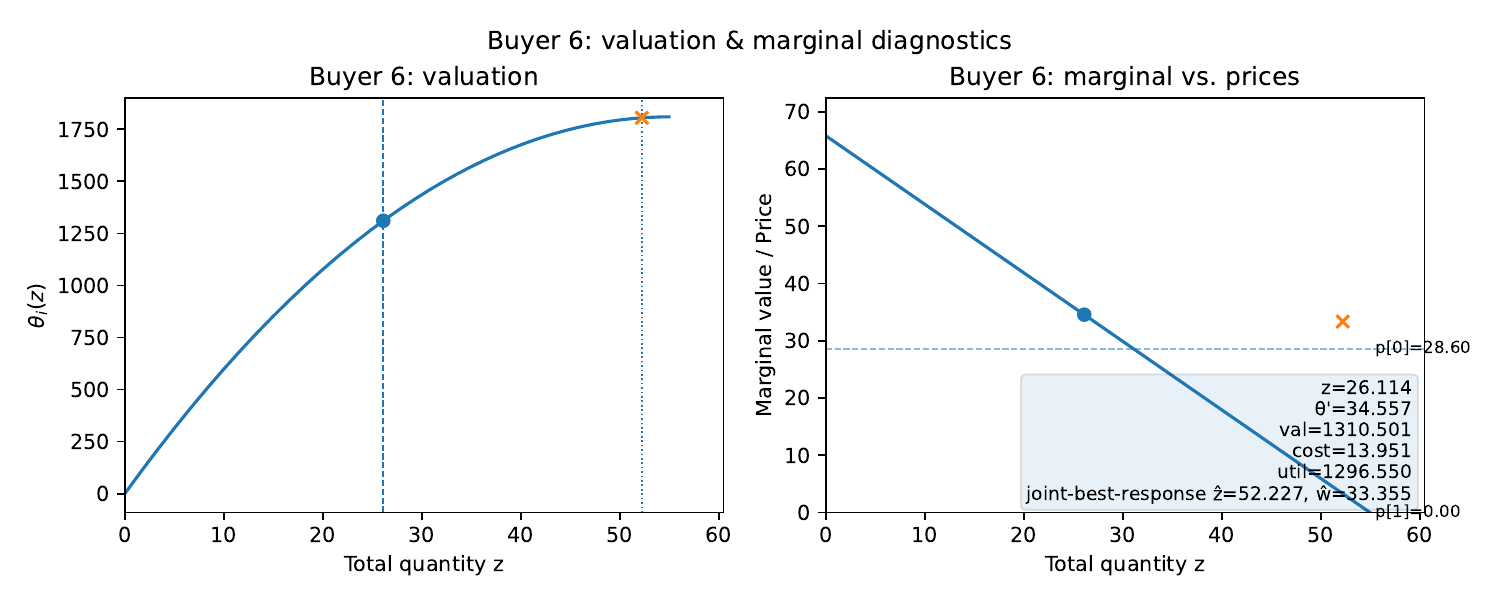} % <-- replace filename
    \caption{Constraint-limited: the joint best response lies on a feasibility boundary; $\theta_i'(Z_i)>p^*$ so the buyer would expand if capacity at $p^*$ were available.}
    \label{fig:jbr-constraint}
  \end{subfigure}\hfill
  %
  % Panel B — Price-limited (θ'(ẑ_i) < p*)
  \begin{subfigure}{0.45\textwidth}
    \centering
    \includegraphics[width=\linewidth]{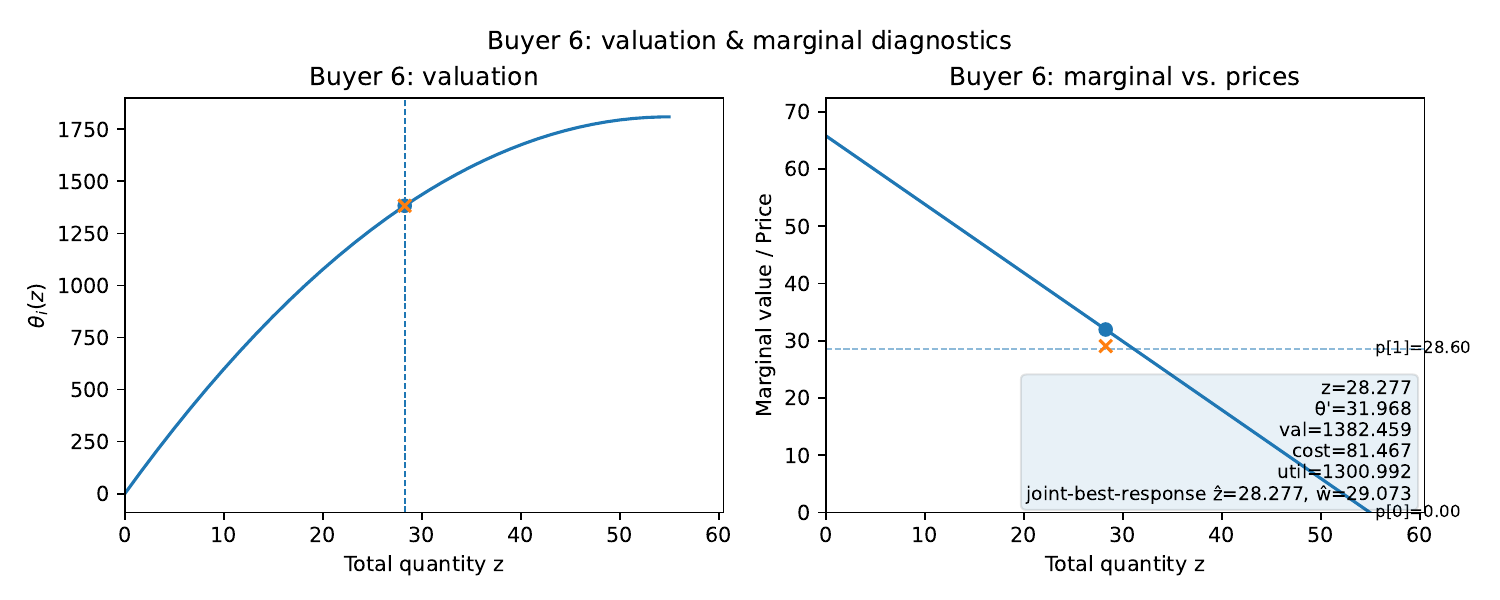} % <-- replace filename
    \caption{Interior optimum: here $\theta_i'(Z_i) > p$, so the buyer would expand if possible.
    The joint–best–response is the point where $\theta_i'(z) = p$.}
    \label{fig:jbr-price}
  \end{subfigure}

  \caption{Buyer-level diagnostics under the Progressive Second Price (PSP) 
  joint best response. Each panel shows the valuation $\theta_i(z)$ with
  the realized point $(Z_i,\theta_i(Z_i))$ and the marginal 
  curve $\theta_i'(z)$ with a dashed line at $p^*$, illustrating the 
  transition from constraint-limited to price-limited behavior along improvement paths.}
  \label{fig:buyer6-diagnostics}
\end{figure*}

When the orange marker in the marginal panel lies above the dashed line, 
the realized point satisfies $\theta_i'(Z_i)>p^*$; the buyer would
buy more at the prevailing price, but the minimal–cost fill has saturated
feasible capacity at that price, so the joint best response is attained 
on a boundary of the feasible region rather than at marginal equality.
When the marker sits on the dashed line, $\theta_i'(Z_i)=p^*$ holds 
and the allocation is locally efficient; here the construction returns an
interior maximizer of $U_i(z)=\theta_i(z)-C_i(z)$. When the marker lies
below the dashed line, $\theta_i'(Z_i)<p^*$ and any further increase
in quantity would decrease utility; the best response is therefore at
or near a participation boundary even though the valuation point on the
left panel is well inside the curve. In every case the left panel places
the realized point on $\theta_i(\cdot)$ because the algorithm maximizes 
value minus payment over the compact feasible set under the current price.

\begin{table}[h!]
\centering
\caption{Buyer regimes and their economic interpretation.}
\begin{tabular}{@{}lll@{}}
\toprule
\textbf{Regime} & \textbf{Relation} & \textbf{Economic meaning} \\ \midrule
{Constraint-limited} & $\theta_i'(z^*) > p^*$ & Supply prohibitive. \\
{Equilibrium (interior)} & $\theta_i'(z^*) = p^*$ & Marginally efficient allocation. \\
{Price-limited} & $\theta_i'(z^*) < p^*$ & Cost prohibitive. \\ \bottomrule
\end{tabular}
\end{table}

Figure~\ref{fig:usurfaceshare} extends the analysis to the joint 
allocation space of the two sellers, 
each point on the surface corresponds to a feasible
distribution across sellers $0$ and $1$ under a uniform price
$w=\theta_i'(Z_i)$.
The height of the surface indicates utility under the split given the opposing bids present in
the snapshot. The ridge along constant $Z_i$ 
identifies the efficient split between sellers:
solutions on the plateau indicate a local optimum; as the solution shifts 
below a ridge the buyer could improve utility by increasing its bid quantity,
and the solution shifts toward the seller facing weaker opposing demand. 
Therefore, even with a uniform price 
tied to $Z_i$, the allocation decision 
remains two-dimensional due to how opponent demand and residual available resource shape the intersection
of feasible pricing and allocations.
\begin{table}[h!]
\centering
\caption{Interpretation of ridges in the buyer's utility surface.}
\begin{tabular}{@{}ll@{}}
\toprule
\textbf{Ridge type} & \textbf{What it corresponds to} \\ \midrule
Sharp rise in $u_i$ & A new seller step becomes available (increase in supply). \\
Sharp drop & Another buyer’s bid dominates $\rightarrow$ PSP second-price step kicks in. \\
Plateau & Both sellers saturated or prices equalize (local equilibrium). \\ \bottomrule
\end{tabular}
\end{table}

\begin{figure}[h!]
    \centerline{\begin{minipage}[b]{0.75\textwidth} 
\caption{\label{fig:usurfaceshare}%
Shared-seller utility surface where buyer~6’s 
utility is a function of total requested quantity $Z_i=z_0+z1$;
$u_i(z_0, z_1) = \theta_i(Z_i) - w(z_o, z_1)(Z_i)$
and $w=\theta'(Z_i)$: the feasible participation surface.
%for buyer~6 with two sellers and a uniform endogenous price
%$w=\theta_i'(z_{\text{total}})$ (no reserve prices in this experiment). The surface is plotted over
%the total requested quantity $z_{\text{total}}$ and split ratio $\alpha=z_{j_0}/(z_{j_0}+z_{j_1})$.
%Peaks along the surface trace efficient splits conditional on opposing demand; symmetry yields
%$\alpha\approx 0.5$, while asymmetry shifts the ridge toward the less-contested seller. The plot
%exposes the two-dimensional nature of the allocation choice even when the price depends only on
%z_{\text{total}}$.
}
    \end{minipage}}
    \centerline{
        \includegraphics[height=0.5\textwidth]{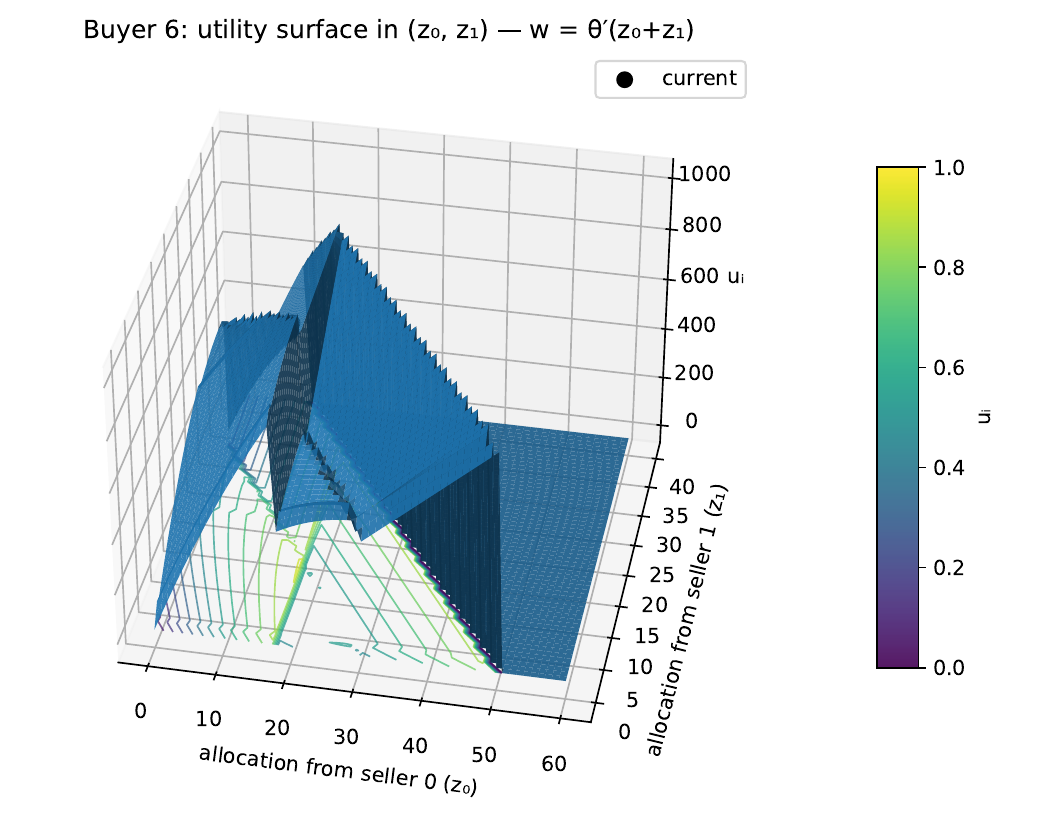}
    }
\end{figure}

To summarize our results, seller~0, with
greater available resource ($Q_{\max} = 60$), cleared at a slightly higher 
price $p_0^* = 30.084$ than Seller~1 ($p^*_1 = 28.597$). Despite this asymmetry,
both sellers exhibited similar expected revenues ($E_0 = 44.76$, $E_1 = 31.16$)
and low variance, attributed to 40\% of buyers participating in both markets. 
The shared influence among these buyers synchronized seller behavior,
leading to a nearly uniform price surface.

Buyer-level data shows that bridging buyers---particularly buyers~6 and~7---maintained
bids across both sellers with marginal valuations $(32.118, 33.355)$ close to Seller~0’s 
clearing price. Their dual participation enforced cross-market coherence, 
ensuring that no single auction could deviate significantly from the shared
equilibrium. Buyers at the margins (e.g., 1 and 2) contributed to shaping 
intermediate prices, stabilizing the monotone progression across tuples.

Furthermore, at each equilibrium or stopping point, the following quantities were collected,
\[
    E_j=\frac{1}{A_j}\sum_{i\in{\cal B}^j} a_i^j p_i
    \quad\hbox{and}\quad
    V_j=\frac{1}{A_j}\sum_{i\in{\cal B}^j} a_i^j (p_i-E_j)^2
    \quad\hbox{where}\quad
    A_j=\sum_{i\in{\cal B}^j} a_i^j
\]
where $E_j$ is the expected seller revenue, $V_j$ is the 
variance of seller revenue across realizations. In addition, buyer classification 
(winners, zero allocation, opt-out behavior) and network statistics 
(fraction of shared buyers, influence overlap) are collected.

Figure~\ref{fig:Ej} shows that as connectivity increases, both sellers exhibit convergence 
in marginal value and clearing price. Seller~0 maintains a higher marginal 
value throughout due to its larger available resource ($Q_j = 60$), but the gap between 
sellers narrows with increasing overlap, confirming that multi-auction buyers 
mediate price synchronization across markets. The alignment of $E(p_i)$ and 
$p^*$ demonstrates how shared influence accelerates equilibrium formation and 
reduces price dispersion.
\begin{figure}[h!]
  \centering
  \includegraphics[width=0.6\textwidth]{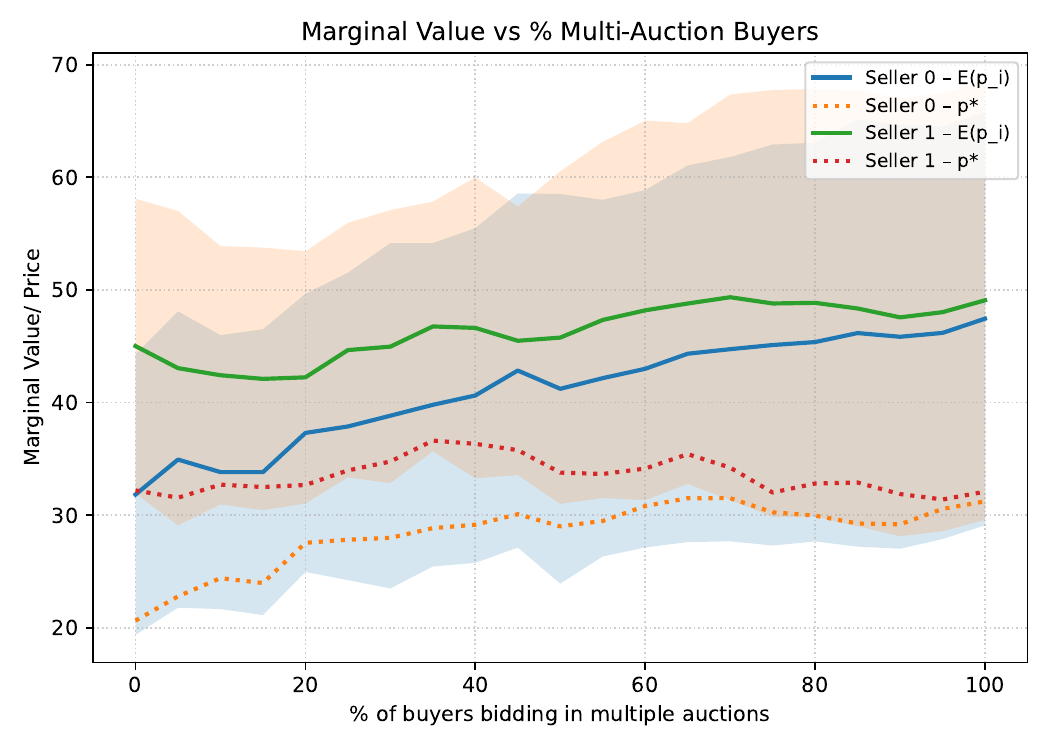}
  \caption{Marginal value versus the percentage of buyers participating in 
  multiple auctions.}
  \label{fig:Ej}
\end{figure}

Overall, increasing connectivity transforms the market from a set of
independent price islands into a unified utility surface. Sparse
configurations produce local equilibria with greater variance between sellers, 
while denser networks encourage faster convergence through influence
propagation. These findings validate the theoretical expectation that
shared influence sets promote global coordination and equilibrium 
alignment.

\paragraph*{Statement on Supplementary Material}
The code for the experiments presented in this paper can be found at:
\begin{itemize}
\item https://github.com/jkblazek/arXiv-2511.19225
\end{itemize}

%%%%%%%%%%%%%%%%%%%%%%%%%%%%%%%%%%%%%%%%%%%%%%%%%%%%%%%%%%%%%%%%%%%%%%%%%%%%
\section{Conclusion and Future Work}\label{sec:conclusion}
%%%%%%%%%%%%%%%%%%%%%%%%%%%%%%%%%%%%%%%%%%%%%%%%%%%%%%%%%%%%%%%%%%%%%%%%%%%%

This paper establishes a graph-theoretic framework for analyzing Progressive Second-Price (PSP) 
auctions, connecting decentralized market dynamics to structural properties of influence graphs.
We formalized and expanded the concepts of influence sets and saturation, which together bound strategy spaces
deterministically and ensure stable, truthful convergence in decentralized settings.

Our analysis relies on two levels of saturation, linked by the partial‐ordering 
structure of bids; local saturation is a set–level best–response property.
Establishing a formal fixed-point characterization of this process--perhaps using
lattice or order-theoretic methods--remains an important direction for future work.

Our present analysis instead focuses on the constructive evolution of influence shells
and the preservation of local monotonicity. Specifically, our approach demonstrates 
how recursive expansions of influence sets reveal market interactions across successive 
auction rounds. By introducing intra-round resolution via the $\tau_k$ steps,
we provide a finer-grained analytical tool to model the internal dynamics of
auction iterations, clarifying the subtle interactions between buyer strategies 
and seller pricing rules.

The establishment of monotonicity in bid updates via induced partial ordering 
ensures stable, non-oscillatory convergence under realistic, elastic valuation
conditions. Our framework provides a robust method to anticipate market shifts, 
characterize equilibrium thresholds, and ensure consistent propagation of 
influence across dynamic network topologies.

Future research will explore several promising extensions in reserve price 
optimization. Could there be an optimal coordinated reserve vector, 
chosen using buyer feedback, that upholds Lemma~\ref{lem:ladder}?
We start by defining an admissible reserve price
region, where for fixed $t$, the admissible set of reserve profiles is
\[
  R(t) = \big\{p_* \in \mathbb{R}^J : p_*^j \in \big(\overline p^j(t),\, \underline p^j(t)\big)
  \ \forall j,\ \text{and Lemma~\ref{lem:ladder} holds}
  \big\}.
\]

Thus, we may determine the existence of \emph{coordinated reserve profile},
where, 
given seller-side weights $w_j \ge 0$ at time $t$, is defined as
any $p_*^{\mathrm{coord}}(t)\in R(t)$ that maximizes
\[
  \Phi(p_*) = \sum_{j=1}^J w_j Q^j(t)\ p_*^j
\]
over $R(t)$.
We conjecture that at least one
coordinated reserve profile $p_*^{\mathrm{coord}}(t)$ exists for every $\epsilon>0$.
Moreover, every such profile preserves the interval ladder
inequalities~\eqref{eq:interval-ordering} and hence is consistent with
local saturation of primary influence shells.

A formal proof of a best–response property
is beyond the scope of this work at this time. Instead, we provide a sketch 
of the proof that would demonstrate the existence of a joint $\epsilon$-best
reply for a buyer participating in multiple concurrent auctions.

To form a coordinated reply at a common marginal price, we collect the sellers
visible to buyer~$i$ under $s_{-i}$ and compute their prices
at a target marginal value $w_i$. Ordering these sellers by nondecreasing 
price and applying tentative allocation until the requested total is reached yields the 
minimal–payment split across auctions. Denote our buyer-level payment by
$$
  P_i(z,s_{-i}) := \inf\{\,y : Q_i(y,s_{-i}) \ge z\,\}, 
$$
which we call an aggregate price staircase. The cumulative payment is
$$
  C_i(z;s_{-i}) := \int_0^z P_i(\zeta,s_{-i})\,d\zeta .
$$
Knowing by finite-valuation that $P_i(\cdot,s_{-i})$ is bounded, 
nondecreasing, and piecewise constant, this
construction implements exactly the payments returned by PSP at the target
marginal price and, among all feasible potential allocations with the same total quantity, 
attains the minimum payment.

First, we aggregate availability across auctions at a common marginal price by 
$P_i(z,s_{-i})=\inf\{y:Q_i(y,s_{-i})\ge z\}$. Finiteness in the number of bids ensures
boundedness and right–continuity; the plateau condition 
$Q_i(y^-,s_{-i})<z\le Q_i(y,s_{-i})$ characterizes $P_i(z,s_{-i})=y$. 
The cumulative payment $C_i(z)=\int_0^z P_i(\zeta,s_{-i})\,d\zeta$ is
continuous and convex. Consider
$$
  U_i(z; s_{-i}) := \theta_i(z) - C_i(z; s_{-i}).
$$
Realize $Z_i$ by sorting tiers by effective PSP price at level
$p_i=\theta_i'(Z_i)$ and taking quantities until the sum equals 
$Z_i$; PSP returns the same aggregate staircase, hence the same total. 

Finally, integrating 
stochastic perturbations and noise into the PSP framework will deepen the
realism of the model, allowing exploration of market resilience under uncertainty.
Additionally, applying resistance distance~\cite{Barrett2017} to the 
reserve profiles and diffusion-based influence models could yield
deeper insights into influence propagation and market stability.  
Empirical validation through simulation and real-world decentralized applications, 
such as spectrum and bandwidth auctions, will be critical to validate and
refine theoretical predictions and improve practical mechanism design.
%–––––––––––––––––––––––––––––––––––––––––––––––––––––––––

\bibliography{bib}

@phdthesis{Semret2000,
author       = {Nadine Semret},
title        = {Market Mechanisms for Network Resource Sharing},
school       = {Columbia University},
year         = {2000},
type         = {Ph.D. Dissertation}
}

@TECHREPORT{Lazarsemret1998,
title = {Design, Analysis and Simulation of the Progressive Second Price Auction for Network Bandwidth Sharing},
author = {Lazar, Aurel and Semret, Nemo},
year = {1998},
institution = {University Library of Munich, Germany},
type = {Game Theory and Information},
url = {https://EconPapers.repec.org/RePEc:wpa:wuwpga:9809001}
}

@article{vickrey1961,
  author = {Vickrey, William},
  title = {Counterspeculation, Auctions, and Competitive Sealed Tenders},
  journal = {Journal of Finance},
  year = {1961},
  volume = {16},
  number = {1},
  pages = {8-37},
  doi = {10.2307/2977633}
}

@article{Clarke1971,
  author = {Clarke, Edward H.},
  title = {Multipart Pricing of Public Goods},
  journal = {Public Choice},
  year = {1971},
  volume = {11},
  pages = {17-33},
  doi = {10.1007/BF01726210}
}

@article{Groves1973,
  author = {Groves, Theodore},
  title = {Incentives in Teams},
  journal = {Econometrica},
  year = {1973},
  volume = {41},
  number = {4},
  pages = {617-631},
  doi = {10.2307/1914085}
}

@article{Maille2007,
author       = {Patrick Maill{'e} and Bruno Tuffin},
title        = {Analysis of the Progressive Second Price Auction},
journal      = {Telecommunication Systems},
year         = {2007},
volume       = {35},
number       = {3-4},
pages        = {245--258}
}

@inproceedings{Qujia2007,
author       = {Zheng Qu and Yue Jia and Peter E. Caines},
title        = {Quantized Progressive Second Price Mechanisms for Decentralized Resource Allocation},
booktitle    = {Proceedings of the 46th IEEE Conference on Decision and Control (CDC)},
year         = {2007},
pages        = {1475--1482}
}

@article{Qujia2009,
author       = {Zheng Qu and Yue Jia and Peter E. Caines},
title        = {Convergence and Stability of Networked Progressive Second Price Auctions},
journal      = {Automatica},
year         = {2009},
volume       = {45},
number       = {9},
pages        = {2107--2114}
}

@inproceedings{Brandt2008,
author       = {Felix Brandt and Tuomas Sandholm},
title        = {On the Existence of Optimal Auctions in the Presence of Collusion},
booktitle    = {Proceedings of the 9th ACM Conference on Electronic Commerce},
year         = {2008},
pages        = {267--276}
}

@article{Wang2021,
author       = {Wei Wang and Sheng Zhong and Jing Chen},
title        = {Differentially Private Mechanism Design for Spectrum Auctions},
journal      = {IEEE Transactions on Mobile Computing},
year         = {2021},
volume       = {20},
number       = {8},
pages        = {2734--2748}
}

@inproceedings{Aguilera2012,
author       = {Marcos K. Aguilera and Sam Toueg},
title        = {Failure Detection and Consensus in Asynchronous Systems},
booktitle    = {Proceedings of the 2012 International Symposium on Distributed Computing},
year         = {2012},
pages        = {114--128}
}

@book{Lynch1996,
author       = {Nancy A. Lynch},
title        = {Distributed Algorithms},
publisher    = {Morgan Kaufmann},
year         = {1996}
}

@inproceedings{Oki2018,
author       = {Takuya Oki and Yusuke Matsubara and Takahiro Yabe and Naoya Takeishi and Yoshinobu Kawahara},
title        = {Modeling Influence Propagation with Graph Neural Networks},
booktitle    = {Proceedings of the 27th International Joint Conference on Artificial Intelligence (IJCAI)},
year         = {2018},
pages        = {3817--3823}
}

@inproceedings{Toussaint2014TheSO, 
title={The Sphere of Influence Graph: Theory and Applications}, 
author={Godfried T. Toussaint}, 
booktitle = {3rd  International  Conference  on  Information  Technology,  System  \&  Management (ICTSM 2014)},
address = {Abu  Dhabi, UAE},
month   = {May},
year={2014}, 
url={https://api.semanticscholar.org/CorpusID:124361892} 
}

@Article{Blocher2021,
AUTHOR = {Blocher, Jordan and Harris, Frederick C.},
TITLE = {An Equilibrium Analysis of a Secondary Mobile Data-Share Market},
JOURNAL = {Information},
VOLUME = {12},
YEAR = {2021},
NUMBER = {11},
ARTICLE-NUMBER = {434},
URL = {https://www.mdpi.com/2078-2489/12/11/434},
ISSN = {2078-2489},
DOI = {10.3390/info12110434}
}

@Article{Blazek2025,
AUTHOR = {Blazek, Jordana and Olson, Eric and Harris, Frederick C.},
TITLE = {The Effects of Bid Latency in a Progressive Second Price Auction},
JOURNAL = {Preprint},
VOLUME = {},
YEAR = {2025},
NUMBER = {},
ARTICLE-NUMBER = {arXiv:submit/7009332},
URL = {},
ISSN = {},
DOI = {}
}

@book{Kleinberg2007,
  author       = {Jon Kleinberg and Éva Tardos},
  title        = {Algorithm Design},
  year         = {2007},
  publisher    = {Addison-Wesley},
  address      = {Boston, MA},
  isbn         = {978-0321295354},
}

@article{Osvaldo2017,
author = {Osvaldo Anacleto and Catriona Queen},
title = {{Dynamic Chain Graph Models for Time Series Network Data}},
volume = {12},
journal = {Bayesian Analysis},
number = {2},
publisher = {International Society for Bayesian Analysis},
pages = {491 -- 509},
year = {2017},
doi = {10.1214/16-BA1010},
URL = {https://doi.org/10.1214/16-BA1010}
}

@article{Baur2012,
author = {Baur, Melanie},
year = {2012},
month = {01},
pages = {},
journal = {X},
title = {Combinatorial Concepts and Algorithms for Drawing Planar Graphs}
}

@misc{Barrett2017,
  doi = {10.48550/ARXIV.1712.05883},
  url = {https://arxiv.org/abs/1712.05883},
  author = {Barrett, Wayne and Evans, Emily J. and Francis, Amanda E.},
  title = {Resistance distance in straight linear 2-trees},
  publisher = {arXiv},
  year = {2017}
}

@article{Quint1994,
title={Sphere of Influence Graphs: Edge Density and Clique Size},
author={Michael, T.S. and Quint, T.},
journal={Mathematical and Computer Modelling},
volume={20},
number={7},
pages={19--24},
year={1994},
publisher={Elsevier},
url = {https://doi.org/10.1016/0895-7177(94)90067-1},
doi = {10.1016/0895-7177(94)90067-1}
}

@article{Shah2011message,
  title={Message-passing in stochastic processing networks},
  author={Shah, Devavrat},
  journal={Laboratory for Information and Decision Systems, Department of EECS, Massachusetts Institute of Technology},
  year={2011},
  url={http://www.mit.edu/~devavrat/papers/message_passing_survey.pdf}
}

@inproceedings{Maille2004,
author = {Delenda, A. and Maillé, Patrick and Tuffin, Bruno},
booktitle = {X},
year = {2004},
month = {01},
pages = {755- 759 Vol.2},
title = {Reserve price in progressive second price auctions},
isbn = {0-7803-8623-X},
doi = {10.1109/ISCC.2004.1358631}
}

\begin{appendix}

\section{Market Shift Revealed by Partial Ordering}\label{appendix:example}

\begin{example}[Market Shift Revealed by Partial Ordering]\label{ex:market_shift}

In this example we model a simple reactive reserve update,
$$  p_*^j(t+1)=\max\ \Big\{
    p_*^j(t),\ \max_{i\notin\mcB^j(t)} p_i^j(t) + \epsilon\Big\},
$$
with $\epsilon>0$ providing strict improvement for convergence.
Thus the seller always keeps its internal bid strictly above the
highest losing buyer, and the reserve price is nondecreasing in $t$.
Alternative
clearing–price rules that set $p_*^j$ to the threshold $\chi^j(t)$ are equivalent
for our results.

Consider a PSP auction market with two sellers $L_1, L_2$ and four buyers $B_3, B_4, B_5, B_6$. Initial buyer-seller connections are represented by the adjacency matrix:

$$
\mbfA^{(0)} = 
\begin{array}{c|cc}
 & L_1 & L_2 \\
\hline
B_3 & 1 & 0 \\
B_4 & 1 & 1 \\
B_5 & 1 & 1 \\
B_6 & 0 & 1 \\
\end{array}
$$

\vspace{0.3em}\noindent
\textbf{Auction Iteration $t=1$}

\textbf{Seller $L_1$} initially receives bids from $B_3, B_4, B_5$. Suppose initial bid prices are ordered as follows:
$$
p_{B_3}^{L_1}(1)=2.0 > p_{B_4}^{L_1}(1)=1.8 > p_{B_5}^{L_1}(1)=1.5.
$$

Progressive allocation steps for $L_1$:
\begin{itemize}
\setlength{\itemindent}{2em}
\item[$\tau_1$]: $B_3$ allocated requested quantity, pays second-highest price $1.8$. Seller updates reserve to $1.8+\epsilon$.
\item[$\tau_2$]: $B_4$ allocated next available quantity, pays $1.5$. Reserve updates to $1.5+\epsilon$.
\item[$\tau_3$]: $B_5$ receives allocation, pays reserve $(1.5+\epsilon)$.
\end{itemize}

\textbf{Seller $L_2$} has bids from $B_4, B_5, B_6$, with initial ordering:
$$
p_{B_5}^{L_2}(1)=1.9 > p_{B_4}^{L_2}(1)=1.7 > p_{B_6}^{L_2}(1)=1.4.
$$

Progressive allocation steps for $L_2$:
\begin{itemize}
\setlength{\itemindent}{2em}
\item[$\tau_1$]: $B_5$ allocated, pays second-highest price $1.7$. Reserve updates to $1.7+\epsilon$.
\item[$\tau_2$]: $B_4$ allocated next, pays $1.4$. Reserve updates to $1.4+\epsilon$.
\item[$\tau_3$]: $B_6$ allocated, pays reserve $(1.4+\epsilon)$.
\end{itemize}

\vspace{0.3em}\noindent
\textbf{Partial Ordering and Initial Influence Sets}
Initially, influence projections:
$$
\pi \circ \varpi^{-1}(L_1)=\{B_3,B_4,B_5\},\quad \pi \circ \varpi^{-1}(L_2)=\{B_4,B_5,B_6\}
$$
Both sellers share buyers $B_4,B_5$, forming an integrated influence structure.

\vspace{0.3em}\noindent
\textbf{Market Shift at $t=2$: Buyer $B_4$ increases bid on $L_1$}
Buyer $B_4$ increases their bid on seller $L_1$ to overtake $B_3$:
$$
p_{B_4}^{L_1}(2) = p_{B_3}^{L_1}(1)+\delta, \quad \delta>0.
$$

This triggers an immediate, asynchronous allocation decision at seller $L_1$:
\begin{itemize}
\setlength{\itemindent}{2em}
\item[$\tau_1$]: Seller $L_1$ allocates to the new highest bidder $B_4$, who pays the second-highest bid $p_{B_3}^{L_1}(1)$. Reserve updates accordingly.
\end{itemize}
The new partial ordering on $L_1$:
$$
p_{B_4}^{L_1}(2) > p_{B_3}^{L_1}(2) > p_{B_5}^{L_1}(2).
$$

\vspace{0.3em}\noindent
\textbf{Coupled Buyer Rebid}
Buyer $B_4$ observes this new allocation outcome and immediately updates their residual demand. Since buyers maintain consistent bid strategies across all sellers, buyer $B_4$ must now adjust their bid quantity for seller $L_2$ simultaneously:
$$
\sigma_{B_4}^{L_2}(\tau_2) = Q_{B_4}(2) - a_{B_4}^{L_1}(\tau_1),
$$
and submits this updated bid quantity at price $p_{B_4}^{L_2}(2)$.

Seller $L_2$, asynchronously and independently from $L_1$, now processes this rebid at its next local step:
\begin{itemize}
\setlength{\itemindent}{2em}
\item[$\tau_2$]: Seller $L_2$ allocates quantity to buyer $B_4$, charging the next-highest price among competing bidders (e.g., buyer $B_5$'s previous bid).
\end{itemize}

\vspace{0.3em}\noindent
\textbf{Propagation of Influence via Projection Mappings:}
The shift at $L_1$ updates the projection and partial ordering structure, immediately affecting the shared buyer set with seller $L_2$. The updated projections remain:
\begin{align*}
\pi \circ \varpi^{-1}(L_1)&={B_3,B_4,B_5}, \\
\pi \circ \varpi^{-1}(L_2)&={B_4,B_5,B_6},
\end{align*}
\noindent but buyer $B_4$'s strategic rebid triggers a recomputation of reserve prices and rebidding decisions at $L_2$, influencing buyer allocations in subsequent $\tau_k$ steps.

Thus, a local change in buyer $B_4$'s bid on one seller ($L_1$) creates a cascading effect through the partial ordering structure, inducing market shifts and influencing allocation outcomes on another seller ($L_2$). The explicit recomputation of partial orderings demonstrates clearly how strategic perturbations propagate through interconnected auction markets.

\end{example}

%–––––––––––––––––––––––––––––––––––––––––––––––––––––––––
%\section{Contraction and Stability on a Finite Network}\label{appendix:contraction}

%\input{TODO/convergence}

\end{appendix}

\end{document}